%% file: main_arxiv_derived_from_final.tex
\documentclass[a4paper,twocolumn]{article}

\AtBeginDocument{%
  }

\usepackage[margin=0.70in]{geometry}
\usepackage{url}
\usepackage{amsmath,amsfonts,amssymb,amsthm}
\newtheorem{definition}{Definition}
\newtheorem{corollary}{Corollary}
\newtheorem{lemma}{Lemma}
\newtheorem{theorem}{Theorem}

\usepackage{booktabs}
\usepackage[textsize=tiny,disable
]{todonotes}
\usepackage{xspace}
\usepackage{algorithm}
\usepackage[noend]{algpseudocode}
\usepackage{amsmath}
\usepackage[capitalize,noabbrev]{cleveref}
\usepackage{pifont}
\usepackage{bbm}
\usepackage{subcaption}
\usepackage{multirow}
\usepackage{siunitx}
\usepackage{latexsym}
\usepackage{adjustbox}
\usepackage{titling}

\DeclareMathOperator*{\argmax}{argmax}
\DeclareMathOperator*{\argmin}{argmin}

\renewcommand{\paragraph}[1]{%
  \smallskip\noindent
  \textbf{#1}
}

\usepackage[symbol]{footmisc}

\begin{document}

\input{commands}

\title{\DPM: Clustering Sensitive Data through Separation}

%

\author{Johannes Liebenow$^*$\thanks{The first two authors contributed equally to the paper.}, Yara Sch\"utt$^{*1}$,\\ Tanja Braun$^\diamond$, Marcel Gehrke$^\ddagger$, Florian Thaeter, Esfandiar Mohammadi$^*$ 
\\ 
\\ ~\\  \phantom{x}\hspace{-0.03em} $*$ Universit\"at zu L\"ubeck, L\"ubeck, Germany
\\     \hspace{0.1em} $\diamond$ Universität M\"unster, M\"unster, Germany
\\     \phantom{x}\hspace{0.7em}$\ddagger$ Universität Hamburg, Hamburg, Germany\\}
\date{}
\thanksmarkseries{arabic}
\maketitle

\begin{abstract}

Clustering is an important tool for data exploration with the goal to subdivide a data set into disjoint clusters that fit well into the underlying data structure. When dealing with sensitive data, privacy-preserving algorithms aim to approximate the non-private baseline while minimising the leakage of sensitive information. State-of-the-art privacy-preserving clustering algorithms tend to output clusters that are good in terms of the standard metrics, inertia, silhouette score, and clustering accuracy, however, the clustering result strongly deviates from the non-private KMeans baseline.

In this work, we present a privacy-preserving clustering algorithm called \DPM that recursively separates a data set into clusters based on a geometrical clustering approach. In addition, \DPM estimates most of the data-dependent hyper-parameters in a privacy-preserving way. We prove that \DPM preserves Differential Privacy and analyse the utility guarantees of \DPM. Finally, we conduct an extensive empirical evaluation for synthetic and real-life data sets. 
We show that \DPM achieves state-of-the-art utility on the standard clustering metrics and yields a clustering result much closer to that of the popular non-private KMeans algorithm without requiring the number of classes.

\end{abstract}

\section{Introduction}
In the field of unsupervised learning, clustering describes the task of subdividing a data set into several disjoint clusters, which is important for data exploration where the goal is to obtain, e.g.,
the expected number of distinct classes.  If a data set contains sensitive information, the clustering algorithm has to be privacy-preserving. There is a rich body of literature on privacy-preserving clustering algorithms~\cite{balkan17,DPCl1Cluster16,DPCl1MultCl17,DPClConstMult18,DPClEasyIns21,DPUtilEffCl21,DPClOptDPWS18,DPClLloyd16,DPClPBDistContour18,DPClStabAssump20,DPClTightApprox20,DPMaxCover21,emmc21,lshsplits2021} that uses the notion of Differential Privacy (DP)~\cite{DwoRo14}. Differentially private algorithms have a privacy budget that determines the information leakage. A small privacy budget stands for a high level of privacy protection, but it also entails a severe utility loss. 
\begin{figure}[h]
    \centering
    \includegraphics[width=0.9\columnwidth]{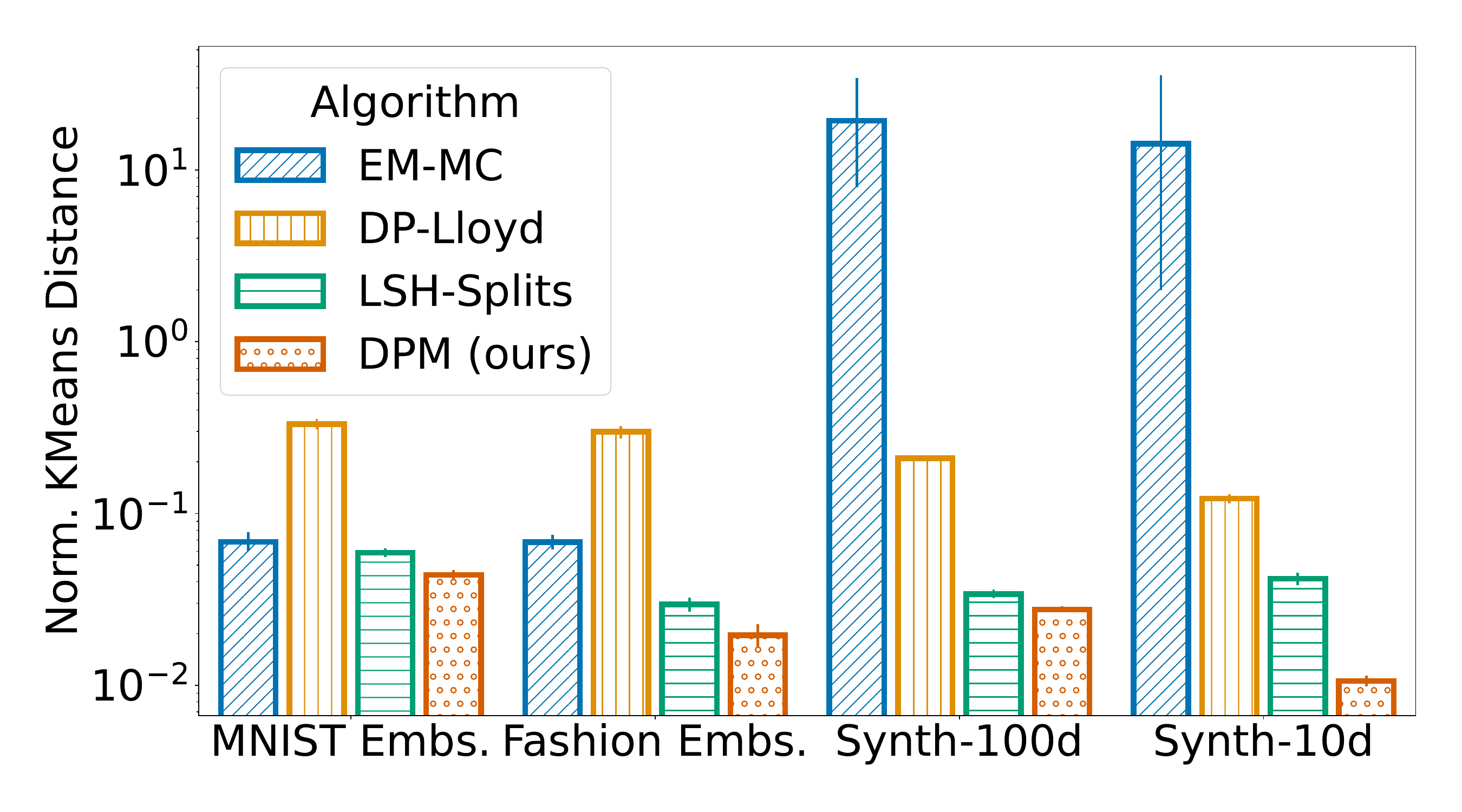}
    \caption{The distance ($\downarrow$) between the clustering result of privacy-preserving clustering algorithms and the non-private KMeans. Our proposed algorithm \DPM outputs a clustering result that is close to that of the non-private baseline. Except for \DPM, all algorithms received the reported number of classes as input for the number of cluster centres. All algorithms use a privacy budget of $\varepsilon = 1$ and $\delta = 1/(n \cdot \sqrt{n})$ and the values are taken from \Cref{tab:kopt_main}.}
    \label{fig:autoKMeansSim}
\end{figure}

We observe that privacy-preserving clustering algorithms tend to output a clustering result that strongly deviates from the non-private KMeans. The reason is that most cluster centres do not contribute much to the clustering result because they are far away from the data points. 
However, the standard metrics, inertia, silhouette score, and clustering accuracy, fail to measure this discrepancy because they are essentially agnostic to centres that are far away from all data points. Thus, a clustering result with only a handful of useful clusters can still have a high quality according to these metrics. But, as part of data exploration, such a deviation from the non-private baseline KMeans can drastically impact decisions on how to further process the data set.

Therefore, we propose a metric called KMeans distance that measures the distance from a private clustering result to the clustering result of the KMeans algorithm and present a privacy-preserving clustering algorithm to minimise this distance. We call our clustering algorithm DP-Mondrian (\DPM) which is based on algorithms that take a geometrical approach to subdivide the data set into clusters like Mondrian~\cite{mondrian06} or Optigrid~\cite{Optigrid99}. Thereby, \DPM does not cluster the data set directly but rather searches for sparse regions along each dimension. When splitting in sparse regions in the data set, structures like clusters will likely be preserved. Thereby, \DPM takes an important step towards being hyper-parameter free by estimating most of its data-dependent parameters in a privacy-preserving way. In summary, \DPM yields a clustering result that is more similar to the output of a non-private baseline KMeans compared to the output of state-of-the-art privacy-preserving clustering algorithms. 
The code is publicly available\footnote{\url{https://github.com/UzL-PrivSec/dp-mondrian-clustering}}.

\subsection{Contribution}
\begin{enumerate}
    \item We present the differentially private clustering algorithm \DPM based on a dimension-wise separation of the data set at regions with high sparseness while separating as close to the median as possible. In particular, \DPM does not require a hyper-parameter search over the number of clusters.
    \item Towards avoiding hyper-parameters in \DPM, we introduce a differentially private algorithm for estimating the granularity of the analysed separation regions, provide an analytical discussion of balancing various optimisation terms, and experimentally show for other parameters that \DPM's utility performance is insensitive to their choice.
    \item We argue that standard clustering metrics do not adequately capture the clustering quality and propose an additional metric, the KMeans distance, to measure the mean distance of a clustering to that of the non-private KMeans. 
    \item We give formal utility guarantees to show the advantage of a geometrical clustering approach to preserve the underlying structure of the data. 
    We conduct an extensive evaluation on synthetic as well as real-world data sets and show that \DPM achieves state-of-the-art performance on standard clustering metrics. In addition, \DPM produces a clustering result closer to that of the non-private KMeans algorithm.
\end{enumerate}

\subsection{Structure}
The remainder of this work is structured as follows. \Cref{sec:prelims} introduces our notation and the topics clustering as well as Differential Privacy. In \Cref{sec:dp-mon} we present the used methodologies and our algorithm \DPM. In \Cref{sec:privacy} we prove that \DPM preserves DP and in \Cref{sec:utility_proof_DPM2} we give utility guarantees for \DPM. In \Cref{sec:eval} we provide an empirical evaluation and in \Cref{sec:related_work} we give an overview of related work. Finally, in \Cref{sec:conclusion} we conclude.

\section{Preliminaries}
\label{sec:prelims}
In this section, we introduce some notation and the field of clustering as well as Differential Privacy (DP).

\subsection{Notation}
We are given a data set, which consists of $n$ $d$-dimensional data points $D = \{x_0, \dots, x_{n-1}\} \in \mathbb{R}^{d \times n}$. The set of all data sets $\datasets$ is defined as the conjunction of all possible data sets: $\datasets = \{D \in \mathbb{R}^{d \times n}|~ d,n \in \mathbb{N}_+\}$. In abuse of notation, we write $S = \{x_{(0)}, \dots, x_{(n-1)}\}$ to indicate an ordered $1$-dimensional set where the type of ordering (ascending or descending) becomes clear from the context. The rank of an element $\wdw$ regarding $D = \{x_{(0)}, \dots, x_{(n-1)}\}$ is defined as the index of $\wdw$ if sorted into $D$ and denoted by $\rankF{\wdw}{D}$.
We write $[z]$ for some $z \in \mathbb{N}$ to indicate the set of natural numbers $\{0, \dots, z-1\}$  When writing $A + B = C$ it means an element-wise addition with $A, B$ being two arbitrary sets of data points. E.g. $A = \{x_1, x_2\}$ and $B = \{y_1, y_2\}$, $A + B = C = \{x_1 + y_1, x_2 + y_2\}$. 
We denote accessing the $i^{\text{th}}$ dimension of data set $D$ as $D^{(i)} = \{ x^{(i)}_0, \dots, x^{(i)}_{n-1} \}$.
$D$ has range $R = [a,b] \subseteq \mathbb{R}$ which means that $\forall i \in [d], \forall x \in D: a \leq x^{(i)} \leq b$. The probability density function (pdf) and cumulative density function (cdf) of a random variable $N$ are written as $\text{pdf}_N$ and $\text{cdf}_N$, respectively. In case of a normal distribution, we write $\phi$ and $\Phi$. The notation of $\tilde n$ denotes that noise was added to preserve privacy. We write whp. for "with high probability".

\subsection{Clustering}
The task of clustering is to partition a data set into a set of clusters that fit well into the underlying structure of the data. Each cluster centre typically has a representative and the set of all representatives is called the cluster centres. More formally, given a data set $D$, we aim to find a set of $k$ cluster centres $C = \{c_0, \dots, c_{k-1}\}$. We say that a data point $x \in D$ belongs to cluster $c_i$ or that $c_i$ represents $x$ if $\argmin_{j \in [k]}||x - C_j||_2 = i$. We use $C_i$ to denote the subset of data points in $D$ which are represented by $c_i$. 
Note that the term cluster centres describe the representatives $c_i$ whereas the term clusters describes sets of data points.
We use $k$ to denote the number of cluster centres, which we also call the size of a clustering result. The term mechanism describes an algorithm.

\subsection{Differential Privacy}
We assume that a given data set contains sensitive information and we follow the notion of DP to process the data set in a privacy-preserving way. In general, we assume that multiple individuals contribute sensitive data points to form a data set. Informally speaking, an algorithm satisfies DP if a single data point can be replaced without changing the output by much. In essence, the output of an algorithm should be robust to small changes in the input because then the influence of single individuals is also limited. A key concept of DP is the so called privacy budget $\varepsilon$ because it calibrates the degree of privacy protection. A small privacy budget means high privacy protection but on the other hand it also entails a strong utility loss because the output has to be randomised.
To talk about the privacy guarantees for individuals, in DP we analyse the difference one data point has on the output of a mechanism. Therefore, we need the terms of neighbouring data sets and the sensitivity of a function.
\begin{definition}[Neighbouring data sets]
Two data sets $D_0,D_1 \in \datasets$ are neighbouring, written as if $D_1= D_0 \cup \{x\}$ for a point $x$.
\end{definition}
\begin{definition}[Sensitivity]
\label{def:sensitivity}
    Given two neighbouring data sets $D_0, D_1$, an arbitrary set $X$, a function $f:\datasets \times X \rightarrow \mathbb{R}$ has sensitivity $\Delta_\score$ iff.\ $\Delta_\score \ge \max_{D_0 \sim D_1} |f(D_0) - f(D_1)|$. 
    A function $f$ with sensitivity $\Delta_\score \in \mathbb{R}$ is a $\Delta_\score$-bounded query.
\end{definition}
\begin{definition}[$(\varepsilon, \delta)$-DP]
    An algorithm $M : \datasets \rightarrow A$ preserves \emph{$(\varepsilon,\delta)$-Differential Privacy (short: $(\varepsilon,\delta)$-DP)} for some $\varepsilon > 0$ and $0 < \delta \leq 1$ if for all $D_0,D_1\in \datasets$ with $D_0 \sim D_1$ and $O\subseteq A$:
    \[\textstyle\Pr[{M(D_0) \in O}] \le \exp(\varepsilon) \textstyle\Pr[{M(D_1) \in O}] + \delta\]
\end{definition} 
To prove that an algorithm preserves DP, we make use of sequential composition, parallel composition, and post-processing. Sequential composition is the property that the composition of $j$ mechanisms that each preserve $(\varepsilon, \delta)$-DP, preserves $(j\varepsilon, j\delta)$-DP. Parallel composition can be applied when using multiple mechanisms that all preserve $(\varepsilon,\delta)$-DP on disjoint subsets of the data set. Then the composition of these mechanisms also preserves $(\varepsilon,\delta)$-DP. Post-processing is a property of DP which guarantees that any algorithm that only processes the output of a $(\varepsilon,\delta)$-DP mechanism, also preserves $(\varepsilon,\delta)$-DP. A common differentially private mechanism is to add noise randomly drawn from a Laplace distribution that is scaled by the sensitivity and the privacy budget.
This mechanism is called the Laplace mechanism and can for instance be used to estimate the number of elements in a set.
\begin{corollary}[DP of Laplace Mechanism \cite{Laplace}]
\label{cor:Laplace-DP}
    Given $\varepsilon \in \mathbb{R}_{>0}$ and a function $f$ with sensitivity $\Delta_\score$, adding noise randomly drawn from $\Lap(\Delta_\score/\varepsilon)$ preserves $(\varepsilon, 0)$-DP.
\end{corollary}
The Exponential Mechanism is an useful tool to select an element that maximises a certain score in a privacy-preserving way. It realises a noisy version of the argmax function.

\begin{definition}[Exponential Mechanism]
\label{def:expMech}
Given an $\Delta_{\score}$-bounded function $\score$:  $\Delta_\score \ge \max_{S\sim S', s\in W_S \cup W_{S'}} |\score(S,\wdw) - \score(S',\wdw)|$, where $S' = S \cup \{x\}$ for some data point $x$. Let $S \in \datasets$ and $\varepsilon >0$, then the Exponential Mechanism $M_{E}$, which takes $S, \score, \varepsilon$ as input draws an element $\wdw \in W$ with probability
\[\mathrm{pmf}_{M_{E}(S,\score,\varepsilon)}(\wdw)=\frac{\exp(\varepsilon\score(S,\wdw)/\Delta_{\score})}{\sum_{\wdw'\in W_S}\exp(\varepsilon\score(S,\wdw')/\Delta_{\score})}\text{,}\]
where $W_S:= dom(\score(S,\cdot))$ is the domain of the second input of $\score$.
\end{definition}
In this work, we consider the following scenario regarding sensitive data: We assume that a data set $D$ contains sensitive information, we further assume the number of dimensions $d$ and the range $R$ of the data set to be given. Every other property of the data set has to be either calculated directly from $d$ and $R$, such as the diameter of the space the data lives in, or has to be computed in a privacy-preserving way. The latter also includes the number of data points $n$ as we will only use its noised version $\tilde{n}$.

\section{\DPM}
\label{sec:dp-mon}
We present the details of our Differential Privacy (DP) clustering algorithm \DPM. First, we present the general approach of \DPM and the scoring function that is used to find reasonable separations. Next, we introduce some algorithmic design choices that are important for \DPM to satisfy DP and present the pseudo-code of \DPM. In the end, we show how to calibrate the scoring function for a given data set in a privacy-preserving way. 

\subsection{Methodology}
\begin{figure}
    \centering
    \includegraphics[scale=0.55]{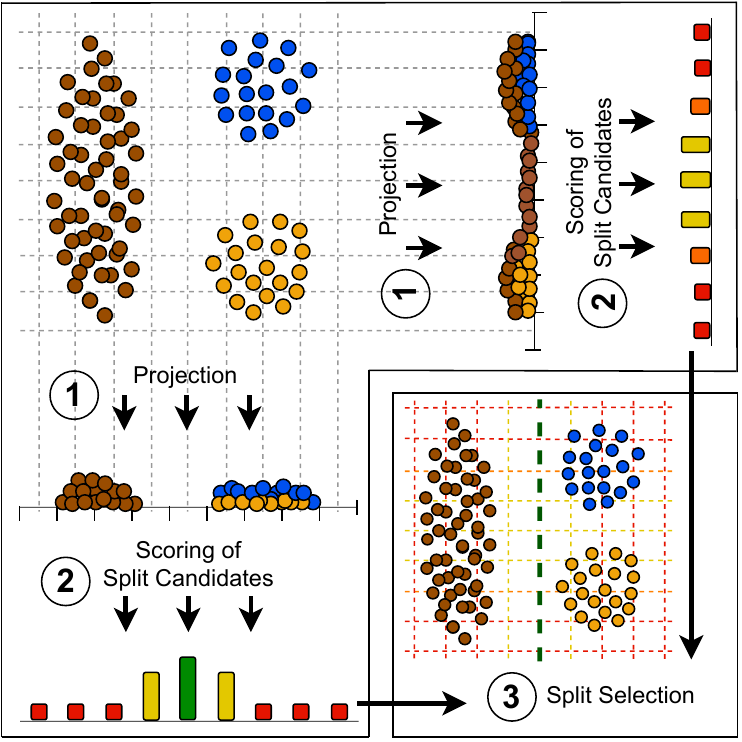}
    \caption{A single recursion step in the process of \DPM. \ding{192} The data points are projected onto each dimension and multiple split candidates are generated based on a fixed split interval size that is calibrated to the data set. \ding{193} A scoring function that depends on the specific clustering goal assigns a score to each split candidate. \ding{194} The split candidate with the highest score gets selected whp. to subdivide the data set into two disjoint subsets. This procedure is  recursively repeated until only a few elements in each subset remain.}
    \label{fig:mondrianClustering}
\end{figure}
\DPM follows a geometrical approach to recursively subdivide a given data set into disjoint clusters that resemble the underlying structure of the data. If the data set contains separate groups of data points, there are regions between these groups that are sparse. \DPM aims to find sparse regions in the data and separate the data set accordingly. The general process is depicted in \Cref{fig:mondrianClustering} and is recursively applied.
To limit the leakage of sensitive information of the overall algorithm, each recursion step of \DPM has to preserve DP (sequential composition) since every recursion step accesses the data set. An important difference between \DPM and other privacy-preserving clustering algorithms is that \DPM does not use a parameter $k$ to determine the intended number of clusters but instead uses an upper limit on the number of splits. 

First, the data points are projected onto each dimension to generate dimension-specific split candidates. Each dimension is divided into equally sized intervals and the centre of each interval is a split candidate. The size of split intervals highly depends on the data and therefore needs to be calibrated in a privacy-preserving way. Then, each split candidate gets a score based on a scoring function that favours splits in sparse regions and simultaneously avoids border regions of the data. Even though border regions are sparse, splits would not result in reasonable cluster centres. Based on these scores, \DPM selects a split candidate with a high score in a privacy-preserving way. To get cluster centres that represent the data points in the data set well, \DPM does not blindly apply splits. A selected split is only applied if there are a minimum number of data points in each subset. To compute cluster centres for each subset that are good representatives, \DPM computes a privacy-preserving average of each subset and outputs the resulting cluster centres.

\subsection{Splits Through Sparse Regions}
\label{Ssec:scoringFunction}
\DPM aims to find splits that separate data in sparse regions. Therefore, \DPM uses a scoring function that consists of two subscores called emptiness and centreness, depicted in \Cref{fig:scoring_func}. The emptiness favours splits in sparse regions and the centreness favours splits close to the median instead of border regions. A weighted combination ($\alpha$ respectively) of both subscores is used as scoring function to balance the subscores. In \cref{ssec:metricWeights_discussion}, we discuss this balance further. Equipped with this compound scoring function, \DPM selects splits in sparse yet relevant regions of the data set. 

The design of the scoring function is not only important in terms of utility but also in terms of DP. Thus, both subscores are designed such that their sensitivity is in $\mathcal{O}(n^{-1})$. Since we work with large data sets, the sensitivity of the scoring function is small and less noise is required to satisfy DP.

\paragraph{Emptiness}
Using the emptiness subscore, splits are favoured that are located in sparse regions of the data. The emptiness measures how many data points are in the vicinity of a split candidate, relative to the total number of data points. Recall that each dimension is divided into evenly spaced intervals of which each centre serves as a split candidate. To determine the emptiness of a split candidate $\wdw$ for a set $S$, \DPM counts the number of elements $|s|$ inside the respective interval, computes the proportion to the number of data points in $S$: $\tilde n$. As we want to maximise the score, we compute the difference between the optimal centreness of $1$ and this proportion.

\begin{figure}[t]
    \centering
    \includegraphics[scale=1.2,trim={0.25cm 0.25cm 0.25cm 0.25cm},clip]{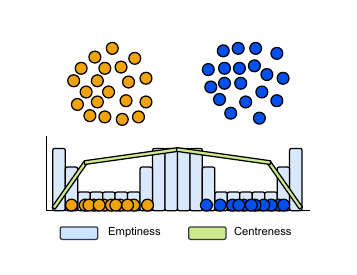}
    \caption{A visualisation of the scoring function that \DPM uses to evaluate split candidates. In order to assign a score, the data points are projected to each dimension. Then, for every split candidate emptiness (light blue) and centreness (light green) are computed.}
    \label{fig:scoring_func}
\end{figure}

\begin{definition}[Emptiness]
\label{def:empt}
Given a split interval size $\wdwSize$, a set $S\subseteq D$, a split candidate $\wdw$ in dimension $i$ and $\tilde n = |S|  + \Lap(\epsCount) \in \mathbb{N}_+$. Then, $|s|$ is the number of data points in $S^{(i)}$ that are contained in the split interval around $\wdw$: $|s| = |\{x \in S^{(i)} \mid s - 0.5\wdwSize \leq x \leq s + 0.5\wdwSize\}|$. Then the \emph{emptiness} $\emptFull{S}{\tilde n}{\wdw}$ is given by 
    \label{def:empt_frac}
    \[\textstyle \emptFull{S}{\tilde n}{\wdw} = 1 - \frac{|s|}{\tilde n}.\]
\end{definition}

The size of split intervals $\wdwSize$ has to be calibrated to the data set to obtain an emptiness that actually favours sparse regions. 

\paragraph{Centreness}
To avoid splitting in border regions of the data, \DPM combines emptiness with centreness. The subscore centreness uses the distance to the median of the data points in the considered dimension to favour splits which are not too far away from the median. To describe the position of split candidates relative to the data points, \DPM uses the concept of ranks. A rank describes the position of a split candidate as the index of that split candidate if it was inserted into the sorted data points (according to the considered dimension).
It is not relevant for us whether a split is right at the median or only close to it as long as the border regions get a bad subscore compared to the centre.
We therefore use a centreness that is composed of two linear functions, an outer one that drastically penalises split candidates far away from the median and an inner one that assigns a similar score to all split candidates in the vicinity of the median. The interval of outer and inner quantiles is determined by $q$ and the minimum score for splits in the inner quantile is $1-t$. Thus, $t$ and $q$, control the position of the transition from the outer to the inner centreness as well as the slope of the linear functions. To ensure that the slope of the inner centreness is smaller, we require $t\ge 2q$.

\begin{definition}[Centreness]
\label{def:cent}
Given parameters $t, q$, a set $S$, split candidate $\wdw$ in dimension $i$, $\tilde n = |S| + \Lap(\epsCount) \in \mathbb{N}_+$, then let $\rank = \rankF{\wdw}{S}$ is the rank of $\wdw$ when inserted in the sorted $S^{(i)}$. We distinguish the cases that $\rank$ is in one of the outer quantiles $\outerQuantile = [0,\tilde nq]\cup[\tilde n-\tilde nq,\tilde n]$ and in the inner quantile $\innerQuantile=(\tilde nq, \tilde n-\tilde nq)$. Then, the \emph{centreness} $\centFull{S}{\tilde n}{\wdw}$ is defined as 
\[
    \textstyle\centFull{S}{\tilde n}{\wdw} = 
        \begin{cases}
            \frac{(\frac{\tilde n}{2} - |\rank-\frac{\tilde n}{2}|) \cdot t}{\tilde nq} & \text{if $\rank \in \outerQuantile$}
            \\ \\
            \frac{t-2q}{1-2q} + \frac{(\frac{\tilde n}{2}-|\rank-\frac{\tilde n}{2}|) \cdot (1-t)}{\frac{\tilde n}{2}-\tilde nq} & \text{if $\rank \in \innerQuantile$}
        \end{cases}\text{.}
\]
We say that $t,q$ are \emph{valid parameters} if $t \ge 2q >0$.
\end{definition}

The scoring function balances the subscores emptiness and centreness with the weight $\alpha$. In \Cref{ssec:metricWeights_discussion} we estimate a useful balance between emptiness and centreness.
\begin{definition}[Scoring function]
\label{def:score}
Given the centreness parameters $t,q$, the split interval size $\wdwSize$, a set $S$, a split candidate $\wdw$ and the subscore weight $\alpha >0$. Then, the \emph{score} $\scoreFull{S}{\tilde n}{\wdw}$ is given by
\[
    \textstyle\scoreFull{S}{\tilde n}{\wdw} = \centFull{S}{\tilde n}{\wdw} + \alpha\emptFull{S}{\tilde n}{\wdw} \text{.}
\]
If $t,q,\wdwSize$ are clear from context we write $\score$. 
\end{definition}

\subsection{Enabling Privacy-Preserving Splits}
\label{sec:priv_splits}

\DPM recursively splits the data set into disjoint subsets and since every recursion step accesses the data set, each step also has to preserve DP. Since the number of data points changes after each step and can become very small compared to the total number of data points, \DPM incorporates some design choices that are necessary to preserve DP and utility. 

\paragraph{Noisy Count}
The noise introduced by the Exponential Mechanism when selecting a split candidate is scaled by the sensitivity of our scoring function, which depends on the number of elements in the current subset. If the sensitivity is less than the true sensitivity, the noise would be insufficient and the privacy guarantees of the Exponential Mechanism would not hold. Since the number of elements $n$ leaks information, we add noise to $n$ using the Laplace mechanism to get $\tilde n$. However, this results in $\tilde n > n$ about $50\%$ of the time. So we introduce an offset $\lambda$ to make sure we underestimate the number of elements: $\tilde n - \lambda \le n$ with probability $1-\deltaCount$. Now, we only overestimate the number of elements and thus violate privacy in the Exponential Mechanism with probability $\deltaCount$. The offset $\lambda$ can be computed as follows:
\begin{align}
    &\mathit{cdf}_{\Lap(1/\epsCount)}(\lambda) \le \deltaCount\nonumber
    \Leftrightarrow 1/2 \exp(-\lambda\epsCount) \le \deltaCount\nonumber\\
     &\Leftrightarrow -\lambda\epsCount \le \ln(2\deltaCount)
    \Leftrightarrow \lambda \ge -\ln(2\deltaCount)/\epsCount \label{eq:offsetEstimation}
\end{align}

\paragraph{Distributing the Privacy Budget}
Distributing the privacy budget is an important design choice when it comes to differentially private mechanisms. We observe that the number of data points in each subset is much higher in the first recursion steps than in the last, which reduces the utility especially in the last step. Therefore, we formulate the following optimisation problem to obtain a privacy budget distribution that minimises the overall noise when assuming an exponentially decreasing number of data points:
\begin{equation}
\begin{aligned}
\min \quad &\textstyle \sum^{\maxdepthRec}_{i=0} \varepsilon_i \cdot 2^{i} \hspace{1em}\textrm{s.t.} \quad & \textstyle \sum^{\maxdepthRec}_{i=0} \varepsilon_i = \varepsilon, \quad \varepsilon_i \geq 0
\end{aligned}
\end{equation}
Using Lagrange Multipliers to solve for $\varepsilon_i$ we obtain $\varepsilon_i = \varepsilon \cdot \sqrt{2^{i}} / \sum_{j=0}^\maxdepthRec \sqrt{2^{i}}$ for a maximum of $\maxdepthRec$ privacy budgets.

\subsection{Clustering Algorithm}

The pseudo-code of \DPM is given in \Cref{algo:dpm}. The input for \DPM consists of the data set $D$, the maximum split depth $\maxdepthRec$, the range of the data set $R$, the centreness parameters $t, q$, candidates for the interval size $\emph{sigmas}$, weight for the emptiness function $\alpha$ and the privacy budgets for each individual step. \DPM starts by initialising the set of cluster centres \emph{clusters} and a set for the noisy cluster sizes \emph{weights}. The total number of data points in $D$ is estimated in a privacy-preserving way which results in $\tilde{n}$. Some steps are repeated for subsets that vary in their size. For these steps, \DPM computes individual privacy budgets. The offsets for calculating the sensitivity of the scoring function are also computed individually for each recursion step. Next, \DPM determines an appropriate size of split intervals $\wdwSize$ through \Call{IntervalSizeEst}{}~(\Cref{algo:wdwsize}). Then, the number of split candidates per dimension $\numSplits$ can be computed based on $\wdwSize$ and the length of the range (Line $1 - 8$). 

\DPM recursively selects and applies splits to subdivide $D$. When the recursion depth $\currDepthRec$ has reached its maximum $\maxdepthRec$, \DPM does not further subdivide a subset. After selecting a split, this split results in two disjoint subsets $S_1, S_2$ with $S_1 \cup S_2 = S$. \DPM checks whether the number of elements in both subsets still satisfy the minimum number of elements $\minNumElem$. To preserve privacy, the number of elements in $S_1$ and $S_2$ is only estimated. Then, if $\tilde{n}_{S_i} \geq \minNumElem, i \in \{1, 2\}$, \DPM continues to recursively subdivide $S_1$ and $S_2$. Otherwise, the selected split is not applied and the recursion stops. (Line $12-19$)

\DPM generates split candidates along each dimension, assigns a score to each, and selects a split with a high score. In detail, \DPM considers $\numSplits$ split candidates per dimension. Split candidates are generated by dividing the range of each dimension $R = [a, b]$ into $\wdwSize$-sized intervals. The centres of these intervals serve as split candidates. The scoring function $\score$~(\Cref{def:score}) used to assess split candidates is introduced in \cref{Ssec:scoringFunction}. \DPM uses the Exponential Mechanism as defined in \cref{def:expMech} (a privacy-preserving argmax) to select the index of the split with a high score. The index of that split is used to derive its dimension $d^*$ and position $s^*$. Using $d^*$ and $s^*$, \DPM subdivides the subset $S$ into $S_1$ and $S_2$ (Line $20-27$).  

After the maximum number of splits $\maxdepthRec$ is reached or the minimum number of elements $\minNumElem$ was not satisfied in a recursion step, a set of clusters \emph{clusters} remains. To obtain the corresponding cluster centres $C$, \DPM applies a privacy-preserving averaging to each cluster (Line $10$). To do this, \DPM adds Gaussian noise scaled to the radius of the data set and the privacy parameters $\epsAverage, \deltaAverage$ to the sum of all elements in a cluster. Then, to obtain the cluster centres, the sums are divided by the noisy number of elements $\tilde n$ in each cluster. Finally, the output of \DPM consists of the noisy cluster centres and the noisy number of elements per cluster.
\begin{algorithm}[t]
\caption{\DPM}\label{algo:dpm}
\begin{algorithmic}[1]
    \Require $D, \maxdepthRec$, $R$, $t$, $q$, \emph{sigmas}, $\alpha$, $\epsWdwSize$, $\epsCount$, $\epsExpMech$, $\epsAverage$
    \State $\text{\emph{clusters}} = \emptyset$
    \State $\text{\emph{weights}} = \emptyset$
    \State $(\epsCounti{i})_{i=0}^{\maxdepthRec} = (\epsCount \sqrt{2^{i}}  / \sum_{j=0}^\maxdepthRec \sqrt{2^{j}})_{i=0}^{\maxdepthRec}$
    \State $(\varepsilon_{\text{exp}_{i}})_{i=0}^{\maxdepthRec-1} = (\epsExpMech \sqrt{2^{i}} / \sum_{j=0}^{\maxdepthRec-1} \sqrt{2^{j}})_{i=0}^{\maxdepthRec-1}$
    \State $(\lambda_i)_{i=0}^{\maxdepthRec} = (-\ln(2\delta)/\epsCounti{i})_{i=0}^{\maxdepthRec}$ \Comment{cf. \Cref{eq:offsetEstimation}}
    \State $\tilde{n} = |D| + \text{Lap}(\epsCounti{0})$
    \State $\wdwSize = \Call{IntervalSizeEst}{D, \tilde n, \epsWdwSize, \emph{sigmas}}$ \Comment{\Cref{algo:wdwsize}}
    \State $\numSplits = (b-a) / \wdwSize$ \Comment{$R = [a, b]$}
    \State $\Call{BuildClustering}{D, \tilde{n}, 0}$
    \State $C = \{ \Call{DPAvg}{C_i, \tilde{n}, \epsAverage} \mid C_i \in \emph{clusters}, \tilde{n} \in  \emph{weights}\}$
    \State \Return $C$,\emph{weights}
    \Procedure{BuildClustering}{$S$, $\tilde{n}$, $\currDepthRec$}
        \State If $\currDepthRec \geq \maxdepthRec$: halt and add $S$ to \emph{clusters} and $\tilde{n}$ to \emph{weights}.
        \State $S_1, S_2 = \Call{Split}{S, \tilde{n}_S, \currDepthRec}$
        \State $\tilde{n}_{S_1} = |S_1| + \text{Lap}(\varepsilon_{\text{cnt}_{\currDepthRec + 1}})$
        \State $\tilde{n}_{S_2} = |S_2| + \text{Lap}(\varepsilon_{\text{cnt}_{\currDepthRec + 1}})$
        \State If $\tilde{n}_{S_i} < \minNumElem$: halt and add $S$ to \emph{clusters} and $\tilde{n}$ to \emph{weights}.
        \State $\Call{BuildClustering}{S_1, \tilde{n}_{S_1}, \currDepthRec + 1}$
        \State $\Call{BuildClustering}{S_2, \tilde{n}_{S_2}, \currDepthRec + 1}$
    \EndProcedure
    \Procedure{Split}{$S$, $\tilde{n}$, $\currDepthRec$}
        \State $\Delta_{\score} = \frac{t/q + \alpha }{\tilde{n} - \lambda_{\currDepthRec}}$
        \State $i^* = M_E(S,\score,\varepsilon_{\text{exp}_{\currDepthRec}})$ \Comment{\Cref{def:expMech} \& \ref{def:score}}
        \State $d^* = \lfloor d \cdot \numSplits / i^* \rfloor$
        \State $s^* = ((d \cdot \numSplits \mod i^*) + 0.5) \cdot \wdwSize$
        \State $S_1 = \{ x \in S \mid x^{(d^*)} \leq s^*\}$
        \State $S_2 = \{ x \in S \mid x^{(d^*)} > s^*\}$
        \State \Return $S_1, S_2$
    \EndProcedure
\end{algorithmic}
\end{algorithm}

\subsection{Split Interval Size Estimation}
\label{ssec:wdwSizeEstimation}

As the emptiness subscore is based on the number of data points in a split interval, the size of split intervals needs to be calibrated to the data set. If the split intervals are too small, the emptiness for split candidates in the centre of a cluster would be quite high. If the split intervals are too large, sparse regions between clusters would not be recognised. Thus, split intervals have to be just large enough such that the emptiness is sensitive to dense regions.

First, let us consider the case of a single distribution, e.g. a Gaussian.
Starting on the outer left or right of the distribution and moving closer to the mean, one observes that the distance between neighbouring data points steadily decreases. The rate of decrease depends on the standard deviation of that distribution. 
We consider the more realistic case of a multidimensional multi-modal distribution, this effect slightly changes. Distances between neighbouring data points only decrease closer to a cluster centre. In this case, the effect is not determined by a single standard deviation but rather by the standard deviations of all clusters. We call this the spread of the data points. To estimate the spread of the data points, we estimate the average distance between neighbouring data points.

Optimally, we could directly estimate the distribution of distances between neighbouring data points. However, this kind of ordered statistics becomes infeasible to compute for $|D| > 1000$~\cite{expQ2020}. Instead, we use quantiles as they can easily be computed and are robust in terms of DP. With quantiles, we can estimate the behaviour of the distances between neighbouring data points.

First, we compute quantiles over distances for distributions created from multiple standard deviation candidates. Then, we derive distances between neighbouring data points for each dimension of the given data set and average these distances. Next, we compute a single quantile over these distances and choose the standard deviation $\sigma^*$ that produces the closest quantile. Finally, we use $\sigma^*$ as baseline and compute the split interval size for \DPM $\wdwSize = \frac{1}{2} \sigma^*$. We choose a factor of $\frac{1}{2}$ because in case of a standard Gaussian $\approx 38\%$ of the data points are within $\frac{1}{2}\sigma$ from the mean.

\paragraph{Algorithm}
The pseudo-code for estimating the split interval size for \DPM is given in \Cref{algo:wdwsize}. First, we create a lookup table $T$ that represents a mapping from a standard deviation $\sigma$ to the $65^{\text{th}}$ percentile $p$ over distances of neighbouring data points $A = \{a_0, \dots, a_{\tilde{n}-2}\}$ with $a_{j} = |x^{(i)}_{(j+1)} - x^{(i)}_{(j)}|$, averaged dimension-wise. The notation $x^{(i)}_{(j)}$ indicates that a data point $x$ has index $j$ along dimension $i$ after sorting. Afterwards, we repeat the process for a given data set $D$ to privately calculate the $65^{\text{th}}$ percentile over distances $\tilde{p}$ by using Algorithm $4.3$ from~\cite{expQ2020}~(DP-PERCENTILE). Finally, we select the standard deviation $\sigma^* = T[p^*]$ which yields the closest percentile $\tilde{p} \approx p^*$ and compute $\wdwSize = \frac{1}{2} \sigma^*$.

\begin{algorithm}[t]
\caption{\textsc{IntervalSizeEst}}\label{algo:wdwsize}
\begin{algorithmic}[1]
\Require $D$, $\tilde n$, $\epsWdwSize$, \emph{sigmas}
    \State $T \gets$ empty mapping
    \For{$\sigma$ in \emph{sigmas}}
        \State Sample $X = \{x_0, \dots, x_{\tilde{n}-1}$\} from $\mathcal{N}(0, \sigma I_d)$ 
        \State $A = \frac{1}{d}\sum_{i=0}^{d-1} \Call{Dists}{X^{(i)}}$ 
        \State $p =$ \Call{percentile}{$A$, $65$} \Comment{Non-private percentile}
        \State $T[p] = \sigma$
    \EndFor
    \State $A = \frac{1}{d}\sum_{i=0}^{d-1} \Call{Dists}{D^{(i)}}$
    \State $\tilde{p} = \Call{dp-percentile}{A, 65, \epsWdwSize}$
    \State $\sigma^* = T[\tilde{p}]$ \Comment{Approximate equality}
    \State \Return $\frac{1}{2} \sigma^*$
\Procedure{Dists}{$x_0, \dots, x_{\tilde{n}-1}$}
    \State $x_{(0)}, \dots, x_{(\tilde{n}-1)} \gets \Call{Sort}{x_0, \dots, x_{\tilde{n}-1}}$ \Comment{Ascending order} 
    \State \Return $(|x_{(i+1)} - x_{(i)}|)_{i=0}^{\tilde{n}-2}$
\EndProcedure
\end{algorithmic}
\end{algorithm}

\subsection{Time Complexity}
\label{ssec:timeComplexityDPM}
To analyse the time complexity of \DPM, we first take a look at the time complexity of \cref{algo:wdwsize} and then analyse \cref{algo:dpm}. 
\cref{algo:wdwsize} starts by drawing and dimensionally sorting the data points for each $\sigma \in \emph{sigmas}$ and all dimensions, resulting in $O(|\emph{sigmas}| \cdot d \cdot |D| \log |D|)$  to compute $A$. After sorting $A$, $p$ can be computed in constant time. 
Now, the process is repeated on $D$ for a single iteration, which again is in $O(d \cdot |D| \log |D|)$. After sorting $A$, it is again possible to compute $\tilde{p}$ in constant time. Assuming a constant lookup in $T$, the time complexity of \cref{algo:wdwsize} is in $O(|sigmas|\cdot d \cdot |D| \log |D| )$.

To analyse the time complexity of \cref{algo:dpm}, we go through each step separately. The computations in lines 1, 2, 6 and 8 are constant and lines 3-5 are linear in $\maxdepthRec$. 
Next, the clustering is obtained, by calling the procedure \textsc{BuildClustering}, which calls the procedure \textsc{Split}. To split the current set $|S|$ into two subsets, first, $\numSplits$ split candidates are generated per dimension.
For each dimension, all the data points in the data set $D$ must be sorted, which is in $O(d \cdot |D| \log|D|  )$. For each set $S$, the desired order of the data points can then be obtained in $O(d\cdot |S|)$.
Given the order, for each split candidate for $S$ the score $\score$ can be computed in $O(|S|)$. Given a split candidate $s^*$, the partitioning of the data into $S_1, S_2$ can also be done in $O(|S|)$. So the procedure \textsc{Split} is in $O(d \cdot |S|)$. 
For all sets $S$ for which the procedure \textsc{BuildClustering} is called  at recursion level $\gamma \le \maxdepthRec$ it holds that: 
\small\[\sum_{S:\text{\textsc{BuildClustering}}(S,\tilde n, \gamma)}d\cdot|S| = d\cdot \sum_{S:\text{\textsc{BuildClustering}}(S,\tilde n, \gamma)}|S|\le d |D|\text{.}\]
Given this inequality, although \textsc{BuildClustering} can be called up to $2^{\maxdepthRec - 1}$ times with a time complexity of $O(d\cdot|S|)$ for each call, the time complexity of all recursion steps is $O(\maxdepthRec \cdot d \cdot |D| + d \cdot |D|\log |D|)$, since \DPM has at most $|D|$ data points to consider for each recursion level. 
Computing the dimension-wise average of all data points in a cluster $C_i$ (Line 10) can be done in $O(d \cdot |C_i|)$. With $|D| = \sum_{C_i\in clusters}|C_i|$, we get $O(d \cdot |D|)$ to compute the cluster centres given a clustering. 
The time complexity of \cref{algo:wdwsize}, dominates the effort to sort the data points before the recursion as well as the effort to compute the cluster centres. Also, the computations in lines 3-5, which are in $O(\maxdepthRec)$, are dominated by the time complexity of the recursive splitting. 
In total, the time complexity of \cref{algo:dpm} is in $O(d \cdot |D| (|sigmas|\log|D| + \maxdepthRec))$.

\section{Privacy Guarantees}
\label{sec:privacy}
To prove that \DPM preserves Differential Privacy (DP), we prove that all parts of \DPM preserve DP.
We start by discussing and proving the DP guarantees of the noisy counts. Then, we prove DP for the interval size estimation and the split selection. Finally, we present the composed result by using sequential and parallel composition theorems.

\subsection{Noisy Count}
\label{Sssec:truncatedNoisedCount}

Several subroutines of \DPM require the number of elements in a set, e.g. the steps split interval size estimation (Line 7), checking whether the stopping criterion is fulfilled (Line 17), computing the score of the split candidates and selecting a split candidate with the Exponential Mechanism (Line 22).  
Since counting elements leaks information about the data set, we add Laplace noise to preserve privacy. The sensitivity of the count is $1$, so with \cref{cor:Laplace-DP} it follows that adding Laplace noise with scale $1/\epsCount$ to the count of elements $|S|$ of a set, preserves $(\epsCount, 0)$-DP.
However, as explained in \cref{sec:priv_splits}, this does not work for the Exponential Mechanism because its sensitivity is based on the number of elements. With a probability of $\approx 50\%$  the number of elements would be overestimated and therefore the sensitivity of the scoring function would be underestimated. 
By introducing an offset $\lambda$ we ensure that we underestimate the number of elements with a probability of $1 - \deltaCount$ to add a sufficient amount of noise. 
Formally, if $\tilde{n} > |S|$, the sensitivity $\Delta_{\score}$ would be underestimated. So we shift the noisy count $\tilde{n}$ by $\lambda$ to get $\tilde{n} -\lambda \le |S|$ whp. $1-\deltaCount$.
If we shift the Laplace distribution by $\lambda \ge -\ln(2\deltaCount)/\epsCount$ (\cref{eq:offsetEstimation}) for some $S\subseteq D\in \datasets$, the bound on the ratio of neighbouring density functions (i.e., the privacy loss bound) of the Laplace mechanism still holds up to $|S|+\lambda-1$. Therefore, we can use the cumulative density function $cdf_{\Lap(1/\epsCount)}$ of the Laplace distribution with scaling factor $1/\epsCount$ to determine the probability of exceeding $|S|+\lambda$: $cdf_{\Lap(1/\epsCount)}(|S| + \lambda - 1)$.
\begin{definition}[Shifted Noised Count]
\label{eq:noisycount}
Given $S \subseteq D \in \datasets, \epsCount,\lambda \in \mathbb{R}_{>0}$, then the shifted noisy count of a set $S$ is: 
    \[\textstyle |S| + \Lap\left(\frac{1}{\epsCount}\right) - \lambda \text{.}\] 
\end{definition}
The following DP guarantees of the shifted noised count of a set $S$ follow from \Cref{eq:offsetEstimation} and the argumentation above.
\begin{corollary}[DP of shifted noisy count]
    The noised count $\tilde{n}_S$ shifted by $-\lambda$ for $\lambda \ge -\ln(2\deltaCount)/\epsCount$ for a subset $S$ with $|S|$ elements preserves $(\epsCount,\deltaCount = 1- cdf_{\Lap(1/\epsCount)}(|S|+\lambda-1))$-DP, where $cdf_{\Lap(1/\epsCount)}(|S|+\lambda-1))$ is the cumulative density function of the Laplace distribution.
\end{corollary}

For the rest of the paper, we write $\tilde{n}_S$ for the noised count of a set $S$ (cf. \cref{eq:noisycount}), and if $S$ is clear from context we write $\tilde{n}$.

\subsection{Interval Size Estimation}
To show that \Cref{algo:wdwsize} is privacy-preserving, note that the only data-dependent parts, and thus sensitive parts, are the number of data points and the percentile calculation. As \DPM already computes noisy counts in each step, we can use the first noisy count $\tilde{n}_D$ i.e. the total number of points as input. 
To compute percentiles privately, we use algorithm $4.3$ from \cite{expQ2020} (DP-PERCENTILE) as a subroutine. The original algorithm takes data points as input and computes the $p^{\text{th}}$ percentiles and in this case the sensitivity is $1$. However, as we operate on distances between neighbouring data points and not on data points directly, this sensitivity analysis does not hold anymore. In the following theorem, we adjust the proof by using a similar argument to show that the sensitivity is $2$.

\begin{lemma}\label{lem:PrivacyPercentile}
    Algorithm DP-PERCENTILE ($4.3$,~\cite{expQ2020}) applied to distances between neighbouring data points has a sensitivity of $2$. 
\end{lemma}
\begin{proof}
    The proof of \cite{expQ2020} considers neighbouring data sets and exchanging a single element can only influence the rank of a single data point by one. When applied to distances between neighbouring data points instead of data points, exchanging a single data point can change two distances in the input, thus the rank of an element in the input can change by a maximum of two.
\end{proof}

The proof of \cref{lem:PrivacyPercentile} is in \cref{Ssec:PrivacyPercentile}.
\begin{lemma}
\label{thm:dp-wdwsize}
    Given $\epsWdwSize > 0$ and a data set $D$, \cref{algo:wdwsize} is ($\epsWdwSize,0$)-differentially private when using the noisy count $\tilde{n}_D$ and algorithm DP-PERCENTILE ($4.3$,~\cite{expQ2020}) to compute the $p^{\text{th}}$ quantile.
\end{lemma}
\begin{proof}
    Using the noisy count $\tilde{n}_D$ from \DPM does not leak additional information, which follows immediately from the post-processing property of DP.
    The subroutine applies the algorithm DP-PERCENTILE with sensitivity $2$ to the distances of neighbouring data points.
\end{proof}

\subsection{Sensitivity of the Scoring Function}
\label{Ssec:score_sen}
To select a split with a high score regarding some scoring function $\score$, a score for every split candidate is computed. 
The differentially private selection mechanism, the Exponential Mechanism as defined in \cref{def:expMech} is used to select a split that whp. has a score. 
The privacy guarantees of the selection depend on the sensitivity of the scoring function.
\subsubsection{Sensitivity of Scoring Function}
The scoring function $\score$ for \DPM not only considers the subscore centreness but also the emptiness of an interval. Therefore, to get the sensitivity $\Delta_\score$ we add up the sensitivity of centreness and emptiness.
First, we analyse the sensitivity of the subscore centreness. If one data point changes, the rank of all other data points in the current set can change. Depending on the rank, different linear functions are used to determine the centreness of the split. Therefore, these cases need to be considered.
\begin{lemma}[Sensitivity of the subscore centreness]
\label{lem:sensitivity_centreness}
   Let $S, S'\in \datasets$ be neighbouring data sets and noisy count $\tilde n$ and offset $\lambda \ge -\ln(2\deltaCount)/\epsCount$ with $\tilde n -\lambda \le |S|,|S'|$. 
   Then, the sensitivity for the centreness with valid parameters $t, q$ defined as in \cref{def:cent}, is $\Delta_{\cent_{t,q}} = \frac{t}{(\tilde{n}-\lambda)q}$.
\end{lemma}
\begin{proof}[Proof sketch]
    The centreness function is a piece-wise linear function that distinguishes the cases of inner and outer quantiles. Thus, to determine the sensitivity it suffices to consider the part with the maximum slope. As $t,q$ are valid for $t\ge 2q$, we only have to consider the difference in the outer quantiles for neighbouring data sets, i.e., split candidate with rank $\rank \in [0,\tilde nq] \cup [\tilde n-\tilde n q, \tilde n]$. Then, we compute the maximum difference for a split candidate in the outer quantiles. As a consequence, $\Delta_{\cent_{t,q}} = \frac{t}{(\tilde{n}-\lambda)q}$ is an upper bound on the sensitivity for the case that  $\tilde n -\lambda \le |S|,|S'|$. The full proof of \cref{lem:sensitivity_centreness} is in \cref{Ssec:PrivacyProofCentreness}.
\end{proof}

Second, for the sensitivity of the subscore emptiness, one data point can change the emptiness of a split candidate for a sets $S$ with $\tilde{n}$ by at most $\frac{1}{\tilde{n}}$ if $x$ is in the interval of $\wdw$; otherwise by $0$.
\begin{lemma}[Sensitivity of the subscore emptiness]\label{lem:sensitivity_emptiness}
    Let $S, S'\in \datasets$ be neighbouring data sets, noisy count $\tilde n$, offset $\lambda \ge -\ln(2\deltaCount)/\epsCount$ with $\tilde n -\lambda \le |S|,|S'|$. 
    Then, the sensitivity of the subscore emptiness for the split interval size $\wdwSize$ defined as in \Cref{def:empt_frac} is $\Delta_\empt = \frac{1}{\tilde{n}-\lambda}$.
\end{lemma}
\begin{proof}[Proof sketch]
    We consider two neighbouring data sets, and we need to distinguish the cases that the additional element is in the split interval of $\wdw$ for $S'$ and that it is not in the split interval of $\wdw$. For both cases, we consider the worst-case difference for neighboring data sets and with $\tilde n -\lambda \le |S|,|S'|$, get $\Delta_\empt = \frac{1}{\tilde{n}-\lambda}$ as an upper bound on the sensitivity. The full proof of \cref{lem:sensitivity_emptiness} is in \cref{Ssec:PrivacyProofEmptiness}.
\end{proof}
Recall that we introduced $\alpha$ as hyper-parameter in \Cref{def:score}. The sensitivity of the emptiness subscore $\empt$ is scaled by $\alpha$.

\begin{lemma}[Sensitivity of the scoring function]\label{lem:sensitivity_scoring}
   Let $S, S'\subseteq D \in \datasets$ be neighbouring data sets, noisy count $\tilde n$ and offset $\lambda \ge -\ln(2\deltaCount)/\epsCount$ with $\tilde n -\lambda \le |S|,|S'|$. 
   With the sensitivity for the subscores centreness (with valid parameters $t, q$ defined as in \cref{def:cent}) and emptiness, and the weight $\alpha$, the sensitivity $\Delta_\score$ for the score of a split candidate in $S$ is given by $\textstyle \Delta_\score \le \frac{t/q + \alpha }{\tilde{n}-\lambda}$.
\end{lemma}
\begin{proof}[Proof sketch]
    The sensitivity of the scoring function is obtained by adding the sensitivity of the subscores emptiness and centreness. As the emptiness is scaled by $\alpha$, so is the sensitivity of the emptiness subscore. Thus, $\textstyle \Delta_\score \le \frac{t/q + \alpha }{\tilde{n}-\lambda}$ is an upper bound on the sensitivity for the case that  $\tilde n -\lambda \le |S|,|S'|$. The full proof of \cref{lem:sensitivity_scoring} is in \cref{ssec:PrivacyProofScore}.
\end{proof}

\subsubsection{Exponential Mechanism}
To select a split that is good regarding our scoring function, the Exponential Mechanism is adapted to the scoring function and the corresponding sensitivity. Then, if the noisy count does not overestimate the number of elements, the privacy guarantees of the Exponential Mechanism hold.

\begin{lemma}[Exponential Mechanism{\cite[Theorem 3.10]{DwoRo14}}]
\label{lem:ExpMech}
For a scoring function $\score$ with sensitivity $\Delta_{\score} \in \mathbb{R}$ and $\epsExpMech > 0$ the Exponential Mechanism is $(\epsExpMech,0)$-DP.
\end{lemma}
\begin{corollary}
\label{cor:expMech_Privacy}
    For the scoring function $\score$ as defined in \cref{def:score} and $t, q$ are valid parameters (defined as in \cref{def:cent}) with sensitivity $\Delta_\score \le \frac{t/q + \alpha }{\tilde{n} - \lambda}$ ~(\cref{lem:sensitivity_scoring}), $\epsExpMech>0$, a set $S$ and the noisy count $\tilde n$ with $\tilde n -\lambda \le |S|$, selecting a split with the Exponential Mechanism is $(\epsExpMech,0)$-DP.
\end{corollary}

We combine the guarantees from \cref{eq:noisycount} and \cref{cor:expMech_Privacy} to get the guarantees for one recursion step of \DPM (Line 12-17).

\subsection{Algorithm}
The privacy proof of \DPM directly follows from the DP guarantees of the subroutines for the noisy count (\cref{eq:noisycount}) as well as the selection process (\cref{cor:expMech_Privacy}) and the composition theorems. 
For every recursion level $\gamma$, all subsets are disjoint, thus every data point is in exactly one subset and has no impact on other subsets of the same recursion level. Therefore, the privacy budget per recursion level $\currDepthRec$ can be used for every subset on this depth without leaking additional information, which is implied by the parallel composition theorem. From the sequential composition theorem follows that any sequence of differentially private mechanisms provides DP.

\begin{lemma}[DP of one recursion depth \DPM]
\label{lem:dpmondrian_one_split} 
Given $\epsCount, \epsExpMech$, $\lambda \in \R_{>0}$ and $S\subseteq D$. \texttt{BuildClustering} preserves $(\epsCount + \epsExpMech, \deltaCount)$-DP with $\deltaCount$ as in \cref{eq:noisycount}.
\end{lemma}
\begin{proof}
    By sequential composition, it follows with \cref{cor:expMech_Privacy} that selecting one split candidate preserves $(\epsCount + \epsExpMech, \deltaCount)$-DP with $\deltaCount = 1 - cdf_{\Lap(1/\epsCount)}(|S|+\lambda)$. As $S$ is subdivided into disjoint subsets $S_l,S_r$, we can apply parallel composition for the recursive call of \textsc{BuildClustering}. For one recursion level $\gamma$, \DPM preserves $(\epsCount + \epsExpMech, \deltaCount)$-DP.
\end{proof}
As we have at most $\maxdepthRec$ recursion levels, the privacy guarantees for obtaining the clusters $clusters$ and the weights immediately follow from the sequential composition theorem.
\begin{corollary}
\label{cor:clusterDP}
    For a fixed maximum recursion depth $\maxdepthRec\in \mathbb{N}_{>0}$, \cref{algo:dpm} preserves $((\maxdepthRec+1) \cdot \epsCount + \maxdepthRec \cdot \epsExpMech), (\maxdepthRec+1) \cdot \deltaCount)$-DP.
\end{corollary}
Finally, we compute a noisy cluster centre for each subset in \emph{clusters} using privacy-preserving averaging (Line 10) and at the beginning \cref{algo:wdwsize} is used to compute the interval size.
\begin{theorem}[Main Privacy Theorem]\label{thm:privacy_dpm2}
    Given $\epsCount, \epsExpMech, \epsAverage,\epsWdwSize \in \R_{>0}, \deltaCount,\deltaAverage \in [0,1]$, and $ n,\lambda \in \mathbb{N}_+$ and a noisy average algorithm \textrm{DP-AVG} that preserves $(\epsAverage,\deltaAverage)$-DP, then \DPM preserves $((\maxdepthRec+1) \cdot \epsCount + \maxdepthRec \cdot \epsExpMech) + \epsAverage + \epsWdwSize, (\maxdepthRec+1) \cdot \deltaCount +\delta_{avg})$-DP.
\end{theorem}
\begin{proof}
    The proof immediately follows from \Cref{cor:clusterDP}, the precondition that the noisy average algorithm is $(\epsAverage,\deltaAverage)$-DP, \cref{thm:dp-wdwsize}, and the sequential composition theorem.
\end{proof}

\section{Utility of \DPM}
\label{sec:utility_proof_DPM2}
This section characterises the cases in which \DPM achieves meaningful guarantees. We show that for reasonable data sets \DPM separates clusters. 
In contrast to prior clustering algorithms (such as KMeans), \DPM does not search for dense areas but rather tries to separate clusters as sparse areas.

In the following we show that by applying the Exponential Mechanism, \DPM selects a good split whp. and that if there is no such split, the algorithm halts.
First, we show that the noisy count is close to the exact count whp. and discuss the balance between the subscores emptiness and centreness. Second, we show that whp. \DPM finds a split with high emptiness and centreness if there are any in the centre. Therefore, we first apply a generalised version of the utility theorem of the Exponential Mechanism.
We subsequently show that if there is a central split $\wdw$ that has a high emptiness, whp. some central split $\wdw'$ with a high emptiness is chosen.
Third, we prove that, if there is no central high-emptiness split and the remaining subset is sufficiently spread, \DPM halts whp.
The utility bounds in this section are parametric in $\tilde n$, which means that if by chance $\tilde n$ has been chosen that is far below $|S|$, utility bounds are worse, but all bounds from \Cref{ssec:dpm_finds_good_splits} and \Cref{ssec:dpm_terminates} hold for any $\tilde n$.
\subsection{Noisy Count Often Close to Real Count}
Recall that the scoring function uses a noisy version $\tilde{n}$ of the real count $|S|$ (of a subset $S$), from which an offset $\lambda$ is subtracted. The shifted noisy count $|S| + \Lap(1/\epsCount) - \lambda$ is close to $|S|$.

\begin{lemma}[Closeness of Noisy Count]
\label{lem:utilityNoisyCount}
    Let $S\subseteq D$ be the considered set and $|S|$ the number of elements in this set and the shifted noisy count $|S| + \Lap(1/\epsCount) - \lambda$. Then, for any $\kappa \ge 0$ we have 
   
    \[\Pr[\left||S| + \Lap(1/\epsCount) - \lambda - |S|\right| > \kappa] \le 1/2\exp(-\kappa \cdot \epsCount) / \delta  \text{.}\]
\end{lemma}

\begin{proof}[Proof sketch]
    The noisy shifted count deviates from the real count by the added Laplace noise and the offset $\lambda$. When plugging in the estimate for $\lambda$ from \cref{eq:offsetEstimation}, with the triangle inequality the lower bound follows. The full proof of \cref{lem:utilityNoisyCount} is in \cref{Ssec:UtilityProofNoisyCount}.
\end{proof}

\paragraph{Asymptotic bounds.} To better understand the guarantees for the noisy count of a subset, we give bounds depending on the parameters $n,\epsCount$ and $\delta$. 
The probability that the noisy count differs by more than $\kappa = \ln(\sqrt{n}/\delta)/\epsCount$, is \begin{align*}
    &\Pr[\big||S| + \Lap(1/\epsCount) - \lambda - |S|\big| \le \ln(\sqrt{n}/\delta)/\epsCount] \\
    &\le 1/2\exp(-\ln(\sqrt{n}/\delta)/\epsCount \cdot \epsCount) / \delta\\
    &= 1/2\exp(-\ln(\sqrt{n}/\delta)) / \delta = 1/2 (\delta /\sqrt{n}) / \delta = 1/(2\sqrt{n}) \text{.}
\end{align*}
For $\delta = 1/(n\sqrt{n})$ (as used in our experiments), the deviation is more than $\kappa = 2\ln(n)/\epsCount$ with probability $1/(2\sqrt{n})$.

For $\kappa = \nu \cdot |S|$, we get that the probability that $\tilde n$ is not within a fixed $\nu$-fraction of $|S|$ exponentially decreases with the number of data points in $S$: $\Pr[\big|\tilde n - |S|\big| \le \nu \cdot |S|] \le 1/2 \exp(-\nu \cdot |S| \cdot \epsCount) / \delta$.

\subsection{Minimal Emptiness \& Emptiness-Weight}
\label{sec:hyper_estimation}
\label{ssec:metricWeights_discussion}
Real data sets, for which reasonable data-independent information about the range is known, exhibit a minimal emptiness $\empt_{\min}$ over all splits, i.e., it does not happen that all data points fall into one interval. As a motivation for our utility characterisations, we show that if the entire data set is distributed like one big Gaussian, the emptiness of all splits is above some minimal emptiness.

The granularity estimation (\Cref{ssec:wdwSizeEstimation}) outputs the size of the split intervals $\wdwSize = \sigma/2$ which is -- informally speaking -- the average over all standard deviations.
First, we analyse the minimal emptiness for the first split. Second, we analyse the minimal emptiness for deeper recursion levels.

We can estimate the lower bound of the emptiness of a split candidate for the first split by selecting $\wdw$ such that the centre of $\wdw$ aligns with the median of the data points (in the considered dimension $d^*$) and let $D$ be a set with data points drawn from one Gaussian (in $d^*$). The worst case, i.e., the fullest split-interval can be found if the mean of the Gaussian is in the centre of an interval. Hence, $\emptFull{D}{\tilde n}{\wdw^*} = 1 - \big(\Phi\big(\frac{\wdwSize}{2\sigma}\big)  - \Phi\big(-\frac{\wdwSize}{2\sigma}\big)\big)$ as the emptiness of a split candidate $\wdw$ of size $\wdwSize$ with centre at the median of the data points (with standard deviation $\sigma$).
Thus for $\wdwSize\approx \sigma/2$, we get the following
\begin{align*}
    \textstyle\emptFull{D}{\tilde n_{D}}{\wdw} &= 1 - \frac{|\wdw|}{\tilde n_D} \ge 1- \left(\Phi\left(\frac{\wdwSize}{2\sigma}\right) - \Phi\left(-\frac{\wdwSize}{2\sigma}\right)\right)\\
    & = 1- \left(\Phi\left(\frac{1}{4}\right) - \Phi\left(-\frac{1}{4}\right)\right) = 0.80258 \text{.}
\end{align*}
We can use this minimal emptiness estimation to find an appropriate configuration for $\alpha$. Two splits realistically differ at most by $0.2$ in their emptiness but by $1$ in their centreness. Therefore, to balance the two subscores emptiness and centreness, we set $\alpha = 5$. With this configuration two split candidates can differ at most by $1$ for both subscores.

For the next split levels, we know that $|\wdw|$ as well as the size of the current subset decreases. As we always split the previous set into two subsets, with $|D|/|S|=2^i$, we approximate the expected minimal emptiness for split level $i$. For split level $i$, we estimate the minimal emptiness as follows $\emptFull{S}{\tilde n_{S}}{\wdw} = 1 - \frac{|\wdw|}{\tilde n_D} \cdot \frac{\tilde n_D}{\tilde n} = 1 - \frac{|\wdw|}{\tilde n_D} \cdot 2^i$.

\subsection{Utilising the Exponential Mechanism}
\DPM uses the Exponential Mechanism to select a split from the set of all split candidates with a score close to the optimal score. The Exponential Mechanism guarantees that for a given scoring function $\score$ and its sensitivity $\Delta_{\score}$, the probability to get a candidate with a score close to the optimal score is bounded. 

We use a generalised version of the Exponential Mechanism's utility theorem which directly follows from the original proof strategy of the Exponential Mechanism. Let $W_{OPT_\omega} := \{\wdw \mid |\scoreF{S}{\tilde n}{\wdw} - OPT(S,\score,W)| \le \omega\}$ be all $\omega$-near optimal splits, i.e., splits with a score of at least $OPT_\omega := OPT(S,\score,W) - \omega$.

\begin{theorem}[Exponential Mechanism's Utility - generalised from~{\cite[Theorem 3.11]{DwoRo14}}]\label{thm:exponential_mechanism_utility}
Fixing a set $S\subseteq D \in \datasets$ and the set of candidates $W$, let $\omega \ge 0$, $OPT(S,\score,W) = \max_{\wdw\in W} \scoreF{S}{\tilde n}{\wdw}\}$ and $W_{OPT_\omega} = \{\wdw \mid |\scoreF{S}{\tilde n}{\wdw} - OPT(S,\score,W)| \le \omega\}$ denote the set of elements in $W$ which up to $\omega$ attain the highest score \break $OPT(S,\score,W)$. Then, for some $\kappa>0$ the Exponential Mechanism $M_E$ satisfies the following property:
\begin{align*}
    &\textstyle\Pr\big[\scoreF{S}{\tilde n}{M_{E}(S,\score,\varepsilon)} \le OPT(S,\score,W) - \omega \\
    &\textstyle \phantom{\Pr\big[} - \frac{2\Delta_{\score}}{\varepsilon}\left(\ln\left( \frac{|W|}{|W_{OPT_\omega}|} \right)+ \kappa \right)\big] \le e^{-\kappa}\text{.}
\end{align*}
\end{theorem}
\begin{proof}[Proof sketch]
    We adapt the proof from \cite[Theorem 3.11]{DwoRo14} such that instead of only considering candidates with the optimal score, we consider the candidates with $OPT(S,\score, W) -\omega$. The full proof of \cref{thm:exponential_mechanism_utility} is in \cref{ssec:utility_generalisedExpMech}.
\end{proof}

\paragraph{Importance of centreness.}
Finding just any split with small emptiness might not suffice because in the border regions there are many split candidates with a high emptiness. Therefore, we aim for splits that actually separate clusters and in this way constitute to a reasonable clustering result.
With the subscore centreness as defined in \Cref{def:cent} splits are preferred that are (per dimension) in the centre.

We adapt the utility guarantees of the Exponential Mechanism to the scoring function~(\Cref{def:score}) and the corresponding sensitivity $\Delta_{\score} = \frac{t/q+\alpha}{\tilde n - \lambda}$~(\Cref{lem:sensitivity_scoring}). 

\begin{corollary}\label{cor:dpmondrian_utility_exp}
Let $S \subseteq D \in \datasets$ with noisy count $\tilde{n}$ and offset $\lambda$ and $\scoreF{S}{\tilde n}{\wdw}$ as defined in \Cref{def:score} be the scoring function of the Exponential Mechanism $M_{E}$ with $\Delta_{\score} =\frac{t/q+\alpha}{\tilde n-\lambda}$. 
$W$ is the set of all possible candidates and $W_{OPT_\omega}$ is the set of all candidates with score above the optimal utility threshold $OPT(S,\score,W) -\omega$. For \DPM, $W$ is the set of all split candidates in $S$. Then, the Exponential Mechanism $M_E$ satisfies the following property:
    \begin{align*}
        &\textstyle\Pr\bigg[\scoreF{S}{\tilde n}{M_{E}(S,\score,\varepsilon)} \le~~OPT(S,\score,W) - \omega\\
        &\textstyle  \phantom{\Pr\big[} -\left(2\frac{\frac{t}{q}+\alpha}{\tilde n -\lambda} \right) \left(\ln\left( \frac{|W|}{|W_{OPT_\omega}|}\right)+ \kappa \right)\bigg] 
         \le e^{-\kappa}\text{.}
    \end{align*}
\end{corollary}

With the addition of centreness, it could in theory happen that the centreness dominates the utility score and bad splits very close to the median are chosen, i.e., split candidates in dense regions, even if there is a split with centreness $\cent \ge t$ and acceptable emptiness. We show that the emptiness of a chosen split is good whp.

\subsection{\DPM Finds Splits With High Emptiness} \label{ssec:dpm_finds_good_splits}

\input{datasets}

\sisetup{
    propagate-math-font = true,
    scientific-notation = false,
    output-exponent-marker = e,
    separate-uncertainty = true,
    retain-zero-uncertainty = true,
    text-series-to-math = true,
    uncertainty-mode = separate,
    input-open-uncertainty = (,
    input-close-uncertainty = ),
    group-separator={,},
    group-digits=integer,
    group-minimum-digits=3,
}

\begin{table*}[t]
    \caption{Clustering quality results in terms of the standard metrics, inertia, silhouette score, and clustering accuracy, as well as the KMeans distance, evaluated on $2$ synthetic and $2$ real-life data sets. Each algorithm receives the reported number of classes as input for the number of classes, except for our proposed algorithm \DPM, which aims to determine the number of clusters on its own. 
    The highlights per row indicate the best-performing algorithm excluding the non-private baseline.}
    \label{tab:kopt_main}
    \centering
\begin{adjustbox}{width=1\textwidth}
\begin{tabular}{@{} @{\hspace*{0.1em}}l @{\hspace*{0.3em}}l @{\hspace*{1.5em}}r @{\hspace*{0.4em}}l @{\hspace{3.5em}}r @{\hspace*{0.4em}}l r @{\hspace*{0.4em}}l r @{\hspace*{0.4em}}l r @{\hspace*{0.4em}}l @{}}
\toprule[1pt]
  &  & \multicolumn{2}{c}{\hspace*{-4em}\textsc{KMeans (non-private)}} &  \multicolumn{2}{c}{\hspace*{-1em}\textsc{DP-Lloyd}}  &         \multicolumn{2}{c}{\hspace*{-1em}\textsc{EM-MC}}  &  \multicolumn{2}{c}{\hspace*{0em}\textsc{LSH-Splits}}  &       \multicolumn{2}{c}{\hspace*{0em}\textsc{DPM (ours)}}  \\
  \cmidrule(l{-0.5em}r{3.5em}){3-4} \cmidrule(l{-0.5em}r){5-6} \cmidrule(lr{0.5em}){7-8} \cmidrule(rl){9-10} \cmidrule(rl){11-12}
  & & mean & \multicolumn{1}{c}{\hspace*{-4.0em}std.} & mean & \multicolumn{1}{c}{\hspace*{-0.0em}std.} & mean & \multicolumn{1}{c}{\hspace*{-2em}std.} & mean & \multicolumn{1}{c}{\hspace*{-0.5em}std.} & mean & \multicolumn{1}{c}{\hspace*{-0.0em}std.} \\
\midrule[1pt]
\multirow{4}{*}{\rotatebox[origin=c]{90}{\textsf{\footnotesize MNIST Embs.}}} & Silh. Score ($\uparrow$) &      \num{0.61} & $\pm$ \num{0.02} & \num{0.14} & $\pm$ \num{0.28} & \num{0.42} & $\pm$ \num{0.07} & \num{0.59} & $\pm$ \num{0.01} & \textbf{\num{0.60}} & \textbf{$\pm$} \textbf{\num{0.01}} \\ 
& Accuracy ($\uparrow$) &  \num{0.98} & $\pm$ \num{0.02} & \num{0.33} & $\pm$ \num{0.12} & \num{0.79} & $\pm$ \num{0.07} & \num{0.95} & $\pm$ \num{0.01} & \textbf{\num{0.98}} & \textbf{$\pm$} \textbf{\num{0.01}} \\
& Inertia ($\downarrow$) &  \num{1.4e+06} & $\pm$ \num{1.3e+05} & \num{7.5e+06} & $\pm$ \num{1.2e+06} & \num{3.8e+06} & $\pm$ \num{7.1e+05} & \num{1.7e+06} & $\pm$ \num{1.3e+05} & \textbf{\num{1.4e+06}} & \textbf{$\pm$} \textbf{\num{6.3e+04}} \\
& Norm. KMeans Dist. ($\downarrow$) & \num{0.00} & $\pm$ \num{0.00} & \num{0.30} & $\pm$ \num{0.05} & \num{0.07} & $\pm$ \num{0.01} & \num{0.03} & $\pm$ \num{0.01} & \textbf{\num{0.02}} & \textbf{$\pm$} \textbf{\num{0.00}} \\
\midrule
\multirow{4}{*}{\rotatebox[origin=c]{90}{\textsf{\footnotesize Fashion Embs.}}} & Silh. Score ($\uparrow$) &  \num{0.40} & $\pm$ \num{0.01} &       \num{-0.10} & $\pm$ \num{0.61} &       \num{0.27} & $\pm$ \num{0.05} &      \num{0.39} & $\pm$ \num{0.01} &  \textbf{\num{0.40}} & \textbf{$\pm$} \textbf{\num{0.01}} \\
& Accuracy ($\uparrow$) &   \num{0.77} & $\pm$ \num{0.04} &       \num{0.22} & $\pm$ \num{0.10} &       \num{0.58} & $\pm$ \num{0.05} &       \textbf{\num{0.75}} & \textbf{$\pm$} \textbf{\num{0.03}} & \num{0.71} & $\pm$ \num{0.03} \\
& Inertia ($\downarrow$) &  \num{1.2e+06} & $\pm$ \num{3.0e+04} & \num{5.7e+06} & $\pm$ \num{1.3e+06} & \num{3.2e+06} & $\pm$ \num{4.5e+05} & \textbf{\num{1.5e+06}} & \textbf{$\pm$} \textbf{\num{5.9e+04}} & \num{1.6e+06} & $\pm$ \num{2.0e+05} \\
& Norm. KMeans Dist. ($\downarrow$) &   \num{0.00} & $\pm$ \num{0.03} &      \num{0.33} & $\pm$ \num{0.04} &       \num{0.07} & $\pm$ \num{0.02} &       \num{0.06} & $\pm$ \num{0.01} & \textbf{\num{0.04}} & \textbf{$\pm$} \textbf{\num{0.00}} \\
\midrule
\multirow{4}{*}{\rotatebox[origin=c]{90}{\textsf{Synth-10d}}} & Silh. Score ($\uparrow$) &  \num{0.97} & $\pm$ \num{0.00} &       \num{0.61} & $\pm$ \num{0.05} &        \num{0.42} & $\pm$ \num{0.05} &       \num{0.93} & $\pm$ \num{0.02} & \textbf{\num{0.96}} & \textbf{$\pm$} \textbf{\num{0.02}} \\
& Accuracy ($\uparrow$) &    \num{1.00} & $\pm$ \num{0.00} &       \num{0.66} & $\pm$ \num{0.05} &        \num{0.49} & $\pm$ \num{0.07} &       \num{0.95} & $\pm$ \num{0.02} & \textbf{\num{0.99}} & \textbf{$\pm$} \textbf{\num{0.01}} \\
& Inertia ($\downarrow$) &  \num{1.0e+06} & $\pm$ 2.0e-10 & \num{6.0e+08} & $\pm$ \num{1.0e+08} &  \num{1.5e+09} & $\pm$ \num{8.4e+07} & \num{1.2e+08} & $\pm$ \num{3.0e+07} & \textbf{\num{1.8e+07}} & \textbf{$\pm$} \textbf{\num{2.7e+07}} \\
& Norm. KMeans Dist. ($\downarrow$) &    \num{0.00} & $\pm$ \num{0.00} &      \num{0.12} & $\pm$ \num{0.01} & \num{14.19} & $\pm$ \num{43.27} &      \num{0.04} & $\pm$ \num{0.01} & \textbf{\num{0.01}} & \textbf{$\pm$} \textbf{\num{0.00}} \\
\midrule
\multirow{4}{*}{\rotatebox[origin=c]{90}{\textsf{Synth-100d}}} & Silh. Score ($\uparrow$) &  \num{0.98} & $\pm$ \num{0.00} &       \num{0.58} & $\pm$ \num{0.04} & 0.09 & $\pm$ \num{0.02}  &       \textbf{\num{0.98}} & \textbf{$\pm$} \textbf{\num{0.01}} & \textbf{\num{0.98}} & \textbf{$\pm$} \textbf{\num{0.01}} \\
& Accuracy ($\uparrow$) &    \num{1.00} & $\pm$ \num{0.00} &       \num{0.63} & $\pm$ \num{0.04} & 0.11 & $\pm$ \num{0.03} &       \textbf{\num{1.00}} & \textbf{$\pm$} \textbf{\num{0.00}}  & \textbf{\num{1.00}} & \textbf{$\pm$} \textbf{\num{0.01}} \\
& Inertia ($\downarrow$) &  \num{1.0e+07} & $\pm$ \num{2.6e-09} & \num{1.5e+10} & $\pm$ \num{9.7e+08} & \num{3.5e+10} & $\pm$ \num{1.5e+09} & \num{8.0e+08} & $\pm$ \num{3.6e+08} & \textbf{\num{5.4e+08}} & \textbf{$\pm$} \textbf{\num{3.4e+08}} \\
& Norm. KMeans Dist. ($\downarrow$) &    \num{0.00} & $\pm$ \num{0.00} &    \num{0.21} & $\pm$ \num{0.01} & 19.32 & $\pm$ \num{26.26} &     \textbf{\num{0.03}} & \textbf{$\pm$} \textbf{\num{0.00}} & \textbf{\num{0.03}} & \textbf{$\pm$} \textbf{\num{0.00}} \\
\bottomrule[1pt]
\end{tabular}
\end{adjustbox}
\end{table*}

As long as the optimal split $\wdw^*$ has an emptiness above a threshold $v$ and is sufficiently central assuming that the data is sufficiently spread, whp. \DPM finds a split with high emptiness. Thus, the centreness does not dominate the decision and \DPM finds a good split whp.
We say a split is $t'$-central in the current subset $S$ if its centreness is at least $t'$, i.e., the $t'$-central splits are $W_{\ge t'}:= \{s \mid \centF{S}{\tilde n}{\wdw} \ge t'\}$ where $\tilde n$ is the approximation of $|S|$ that is used in the respective run with $\tilde n \le |S|$ (whp. cf. \Cref{eq:offsetEstimation}).

\paragraph{Choosing good central splits.}
If a $t'$-central split $\wdw$ is selected, the probability is high that the emptiness of $\wdw$ is close to the optimal split's emptiness.

\begin{lemma}[\DPM finds a good central split]
\label{lem:utility_inner_quantile}
Let $S\subseteq D$ be a set with noisy count $\tilde n$ and offset $\lambda$. 
Let $\wdw^*$ be the split candidate with the optimal score $\scoreF{S}{\tilde n}{\wdw^*} = OPT(S,\score, W) = \max_{\wdw \in W} \scoreF{S}{\tilde n}{\wdw}$ and $W_{OPT_\omega} = \{\wdw \mid |\scoreF{S}{\tilde n}{\wdw} - OPT(S,\score,W)| \le \omega\}$. For any $t' \in [0,1]$, if 
$\wdw^* \in W_{\ge t'}$, then
\begin{align*}
&\textstyle\Pr\big[ \left(\frac{\frac{2t}{q \alpha} + 2 }{(\tilde n - \lambda)\varepsilon} \right) \left(\ln\left( |W|/|W_{OPT_\omega}|\right)+ \kappa + \omega \right) + \frac{1-t'}{\alpha}\\
&\phantom{\textstyle\Pr\big[ } > \emptF{S}{\tilde n}{\wdw^*} -\emptF{S}{\tilde n}{\wdw} \mid \wdw \in W_{\ge t'}\big] \ge 1 - e^{-\kappa} \text{.}
\end{align*}
\end{lemma}
\begin{proof}[Proof sketch]
    To give a lower bound for the probability that a selected central split also has a good emptiness, we can apply \cref{cor:dpmondrian_utility_exp} and simplify the bounds. The full proof of \cref{lem:utility_inner_quantile} is in \cref{Ssec:UtilityProof_innerQuantile}.
\end{proof}

\paragraph{Asymptotic Bounds.} To get a better understanding of the utility guarantees, we give asymptotic bounds in $n, d$ and $\varepsilon$. We know $\frac{2t}{q \alpha} + 2$ to be constant. Further, we know that $\ln(|W|/|W_{OPT_\omega})$ to be in $O(\ln(d))$ where $d$ is the dimension of the input data set. Then, we know that the first factor $\left(\frac{2t/(q \alpha) + 2 }{(\tilde n - \lambda)\varepsilon} \right) \left(\ln\left( |W|/|W_{OPT_\omega|}\right)+ \kappa + \omega \right)$ is in $O(1/(n\varepsilon))$. Thus, the difference between the emptiness of the selected split and the optimal split is asymptotically less than $O(\kappa/(n\varepsilon)) (1-t')/\alpha$ with probability $1-\exp(-\kappa)$. For increasing $n$, the maximum difference converges to $(1-t')/\alpha$. Note that the minimal emptiness decreases with increasing recursion depth. Therefore, even if the deviation bound gets worse with every split as the number of elements decrease, the maximum difference in the emptiness increases as well which allows for more deviation from the optimal emptiness. Also, with every split, the privacy budget is increased. The product of the bounds of \Cref{lem:prob_innerQuantileSelected} and \Cref{lem:utility_inner_quantile} provides a bound on the probability that \DPM draws a split with near-optimal emptiness. Asymptotically, we can already understand from the bounds that the hyper-parameter $\alpha$ needs to be sufficiently high to ensure that the centreness-term $(1-t')/\alpha$ does not dominate the emptiness of the split.

\sisetup{
    propagate-math-font = true,
    scientific-notation = false,
    output-exponent-marker = e,
    separate-uncertainty = true,
    retain-zero-uncertainty = true,
    text-series-to-math = true,
    uncertainty-mode = separate,
    input-open-uncertainty = (,
    input-close-uncertainty = ),
    group-separator={,},
    group-digits=integer,
    group-minimum-digits=3,
}

\begin{table*}[t]
    \caption{Clustering quality results in terms of the standard metrics, inertia, silhouette score, and clustering accuracy, as well as the KMeans distance. All algorithms are evaluated on two real data sets that are difficult to cluster, as indicated by the silhouette score. Each algorithm receives the reported number of classes as input for the number of classes, except for our proposed algorithm \DPM, which aims to determine the number of clusters on its own. 
    The highlights per row indicate the best-performing algorithm excluding the non-private baseline.}
    \label{tab:kopt_ablation}
    \centering
\begin{adjustbox}{width=1\textwidth}
\begin{tabular}{@{} @{\hspace*{0.1em}}l @{\hspace*{0.3em}}l @{\hspace*{1.5em}}r @{\hspace*{0.4em}}l @{\hspace{3.5em}}r @{\hspace*{0.4em}}l r @{\hspace*{0.4em}}l r @{\hspace*{0.4em}}l r @{\hspace*{0.4em}}l @{}}
\toprule[1pt]
  &  & \multicolumn{2}{c}{\hspace*{-4em}\textsc{KMeans (non-private)}} &  \multicolumn{2}{c}{\hspace*{-1em}\textsc{DP-Lloyd}}  &         \multicolumn{2}{c}{\hspace*{-1em}\textsc{EM-MC}}  &  \multicolumn{2}{c}{\hspace*{0em}\textsc{LSH-Splits}}  &       \multicolumn{2}{c}{\hspace*{0em}\textsc{DPM (ours)}}  \\
  \cmidrule(l{-0.5em}r{3.5em}){3-4} \cmidrule(l{-0.5em}r){5-6} \cmidrule(lr{0.5em}){7-8} \cmidrule(rl){9-10} \cmidrule(rl){11-12}
  & & mean & \multicolumn{1}{c}{\hspace*{-4.0em}std.} & mean & \multicolumn{1}{c}{\hspace*{-0.0em}std.} & mean & \multicolumn{1}{c}{\hspace*{-2em}std.} & mean & \multicolumn{1}{c}{\hspace*{-0.5em}std.} & mean & \multicolumn{1}{c}{\hspace*{-0.0em}std.} \\
\midrule[1pt]
\multirow{4}{*}{\rotatebox[origin=c]{90}{\textsf{UCI Letters}}} & Silh. Score ($\uparrow$) &  \num{0.15} & $\pm$ \num{0.00} &       \textbf{\num{0.09}} & \textbf{$\pm$} \textbf{\num{0.02}} &       \textbf{\num{0.09}} & \textbf{$\pm$} \textbf{\num{0.02}} &      \num{0.07} & $\pm$ \num{0.01} &       \num{0.05} & $\pm$ \num{0.01} \\
& Accuracy ($\uparrow$) &    \num{0.28} & $\pm$ \num{0.01} &       \num{0.13} & $\pm$ \num{0.02} &       \num{0.11} & $\pm$ \num{0.02} &      \textbf{\num{0.24}} & \textbf{$\pm$} \textbf{\num{0.02}} &      \num{0.20} & $\pm$ \num{0.02} \\
& Inertia ($\downarrow$) &  \num{5.8e+05} & $\pm$ \num{2.8e+03} & \num{1.1e+06} & $\pm$ \num{6.4e+04} & \num{1.2e+06} & $\pm$ \num{6.7e+04} & \textbf{\num{8.5e+05}} & \textbf{$\pm$} \textbf{\num{2.6e+04}} & \num{9.5e+05} & $\pm$ \num{3.3e+04} \\
& Norm. KMeans Dist. ($\downarrow$) &    \num{0.00} & $\pm$ \num{0.00} &      \num{0.20} & $\pm$ \num{0.01} & \num{23.81} & $\pm$ \num{67.34} &      \textbf{\num{0.07}} & \textbf{$\pm$} \textbf{\num{0.00}} &       \num{0.10} & $\pm$ \num{0.01} \\
\midrule
\multirow{4}{*}{\rotatebox[origin=c]{90}{\textsf{UCI Gas}}} & Silh. Score ($\uparrow$) &  \num{0.37} & $\pm$ \num{0.00} &      -\num{1.00} & $\pm$ \num{0.00} &      -\num{1.00} & $\pm$ \num{0.00} &      \num{-0.08} & $\pm$ \num{0.56} &       \textbf{\num{0.13}} & \textbf{$\pm$} \textbf{\num{0.27}} \\
& Accuracy ($\uparrow$) &  & N/A    &    & N/A &         & N/A &       & N/A &       & N/A \\
& Inertia ($\downarrow$) &  \num{2.5e+07} & $\pm$ \num{1.0e+01} & \num{5.1e+07} & $\pm$ \num{5.7e+06} & \num{4.0e+07} & $\pm$ \num{5.0e+05} & \num{1.0e+08} & $\pm$ \num{2.7e+07} & \textbf{\num{3.7e+07}} & \textbf{$\pm$} \textbf{\num{1.6e+07}} \\
& Norm. KMeans Dist. ($\downarrow$) &    \num{0.00} & $\pm$ \num{0.00} &    \num{0.25} & $\pm$ \num{0.02} &      \textbf{\num{0.00}} & \textbf{$\pm$} \textbf{\num{0.00}} &    \num{0.02} & $\pm$ \num{0.00} &     \num{0.01} & $\pm$ \num{0.00} \\
\bottomrule[1pt]
\end{tabular}   
\end{adjustbox}
\end{table*}

In \Cref{ssec:metricWeights_discussion}, we estimate a minimal emptiness of $\ge 0.8$ for the setting that the distribution of the points in the current subset $S$ follow a Gaussian. Therefore, we first analyse the probability that a selected split is in the inner quantile. Afterwards, assuming the selected split is in the inner quantile, we give the bounds on its emptiness compared to the optimal split.  In other words, we consider the case where the centreness term does not lead \DPM astray.

\paragraph{\DPM chooses a central split.} 
We analysed the probability that if the selected split is central, it has a high emptiness. Thus, we gave bounds for the conditional probability (condition $B$: \DPM chooses a central split). Now, to get the probability that a central split with high emptiness is selected, by the law of total probability, it suffices to consider the case where there is a good split within the $t'$-central split $W_{\ge t'}$: Given two events $A$ and $B$:
\begin{align*}
    &\Pr[A] = \Pr[A \mid B] \cdot \Pr[B]
    + \Pr[A \mid \lnot B] \cdot \Pr[\lnot B]\\
    \Rightarrow & \Pr[A] \ge \Pr[A \mid B] \cdot \Pr[B]
\end{align*}
With the analysis of the probability that 
    ``\DPM chooses a $t'$-central split, i.e., $\wdw\in W_{\ge t'}$.'' (event $B$)
occurs, it suffices to analyse the probability that
``\DPM chooses a split with emptiness larger than a bound $v$'' (event $A$)
under the condition that a $t'$-central split $\wdw$ ($\wdw \in W_{\ge t'}$) was chosen (event $B$). 
Put differently, we consider the case where the centreness term does not lead \DPM astray.
Let $W_{\ge t'}:= \{s \mid \centF{S}{\tilde n}{\wdw} \ge t'\}$ be the set of all central splits in the current subset $S$ for which the centreness is at least $t'$ (for some $\tilde n$).
For any $t'\in[0,1]$, the probability to choose a $t'$-central split $\wdw \in W_{\ge t'}$ is proportional to the number $|W_{\ge t'}|$ of $t'$-central splits in $W_{\ge t'}$ and the spread of the data. We formalise the spreadedness of the data with the minimal emptiness of any split $\wdw \in W$. In \Cref{sec:hyper_estimation} we show that if all data is distributed as one big Gaussian, the emptiness of a central split is at least $0.8$ for the first split.
\begin{lemma}[Central quantile is chosen]
    \label{lem:prob_innerQuantileSelected_eventA}
    Let $S\subseteq D$ be the current subset and $W$ the set of all split candidates and $W_{\ge t'}:= \{s \mid \centF{S}{\tilde n}{\wdw} \ge t'\}$ for any $t'\in [0,1]$.
    Let $e_{\min} := \min_{\wdw \in W} \emptF{S}{\tilde n}{\wdw}$ be the minimal emptiness over all splits $\wdw \in W$.
    The score of every split candidate can be represented as $\alpha \cdot e_{\min} + t' + \ln a_s$ (for some $a_s \ge 1$). Then with $
        L_{\ge t'} = \textstyle\sum_{\wdw \in W_{\ge t'}} a_\wdw 
    $ and $
        \textstyle L_{< t'}  = \sum_{\wdw \in W_{< t'}} a_s
    $, we know
    \begin{align*}
        &\Pr[\wdw \in W_{\ge t'}] = 
        \frac{%
                1%
            }{%
                \frac{%
                    L_{< t'}%
                }{L_{\ge t'}} + 1%
            }
        \text{.}
    \end{align*}
\end{lemma}
\begin{proof}[Proof sketch]
    The probability that a chosen split is in the $t'$-central split, is the sum of the scores of all $t'$-central candidates divided by the sum of all split candidates. This can be simplified to $\frac{1}{L_{<t'}/L_{\ge t'} + 1}$ where the $L$ terms hold the sum of differences to the minimal score of a $t'$-central split candidate. For $L_{<t'}$ the sum is over all not $t'$-central split candidates and for $L_{\ge t'}$ over all $t'$-central split candidates. The full proof of \cref{lem:prob_innerQuantileSelected_eventA} is in \cref{SSec:UtilityProof_innerQuantile_eventA}.
\end{proof}
To approximate this probability, we give an upper bound on $L_{<t'}$ and a lower bound on $L_{\ge t'}$. Let $s^*$ be the optimal split with $\centF{S}{\tilde n}{\wdw^*}\ge t'$ and $\emptF{S}{\tilde n}{\wdw^*} \ge v$, $W_{OPT_\omega}$ be the near-optimal split candidates with score at least $\scoreF{S}{\tilde n}{\wdw^*} - \omega $.
We get an untight approximation that takes into account the number of $t'$ central splits and $B=\varepsilon/(2\Delta_f)$ we get:
\begin{align*}
        &\Pr[M_E(S,\score, \varepsilon)\in W_{\ge t'}] \\
        &\ge \frac{%
                1%
            }{%
                \frac{%
                    e^{(OPT_\omega-t' -\alpha \empt_{\min})B}\left(|W_{< t'}\setminus W_{OPT_\omega}|  + |W_{< t'} \cap W_{OPT_\omega}|e^{\omega B}\right)
                }{|W_{\ge t'}\setminus W_{OPT_\omega}| + | W_{\ge t'}\cap W_{OPT_\omega}| e^{(OPT_\omega -t'-\alpha \empt_{\min})B}} + 1%
            }
\end{align*}
More details on this approximation are in \cref{SSec:UtilityProof_innerQuantile_eventA}.

The untight approximation clearly illustrates that in the worst, a large relative emptiness weight $\alpha$ causes \DPM not to split in a central quantile whp. As a result, \DPM would need more split levels to terminate. To this end, our calibration of the relative emptiness weight $\alpha$ is chosen conservatively (cf. \Cref{sec:hyper_estimation}).
If the data set is distributed as a clearly separable Gaussian mixture with Gaussians that are not too far apart in magnitude, we have one of two effects: either there are $t'$-central splits with high emptiness (and thus high score) or the data points are not too much clumped in one split interval close to the $1-2q$-quantile, i.e., $t'$ is close to $t$. Thus, $L_{\ge t'}$ is far larger than $L_{< t'}$ for a $t'$ close to $t$ and the bound for a split to be chosen from $W_{\ge t'}$ approaches $1$.
As our experiments show that \DPM finds a clustering result that is close to the KMeans clustering result, the chosen splits $\wdw$ cannot have been far away from main bulk of the data, i.e., must have been a $t'$-central split ($\wdw \in W_{\ge t'}$)  for some $t'$ close to $t$.

\subsection{\DPM Terminates Appropriately}\label{ssec:dpm_terminates}
Next, we consider the case that there is no good split in any dimension. We show that whp. the algorithm halts before the maximum recursion depth is reached. The probability that \DPM halts, depends on the ratio $L_{<t'}/L_{\ge t'}$ which is bounded by the ratio of $t'$-central splits and not $t'$-central splits as well as their bounds for the score.
\begin{lemma}[Outer quantile is chosen]
    \label{lem:prob_innerQuantileSelected}
    Let $S\subseteq D$ be the current set with noisy count $\tilde n$ and $W$ the set of all split candidates and $W_{\ge t'}:= \{s \mid \centF{S}{\tilde n}{\wdw} \ge t'\}$.
    Let $e_{\min} := \min_{\wdw \in W} \emptF{S}{\tilde n}{\wdw}$ be the minimal emptiness over all splits $\wdw \in W$. The score of every split candidate can be represented as $\alpha \cdot e_{\min} + t' + \ln a_s$ (for some $a_s \ge 1$). Then with $
        L_{\ge t'} = \textstyle\sum_{\wdw \in W_{\ge t'}} a_\wdw 
    $ and $
        \textstyle L_{< t'}  = \sum_{\wdw \in W_{< t'}} a_s
    $, we know
    \begin{align*}
        &\Pr[\wdw \in W_{< t'}] = 
        \frac{%
                1%
            }{%
                \frac{%
                    L_{< t'}%
                }{L_{\ge t'}} + 1%
            }
        \text{.}
    \end{align*}
\end{lemma}
The proof follows the proof of \cref{lem:prob_innerQuantileSelected_eventA}.The full proof of \cref{lem:prob_innerQuantileSelected} is in \cref{Ssec:UtilityProof_innerQuantileSelected}.
To give an approximation as we consider the counter probability, we have to upper bound $L_{\ge t'}$ and lower bound $L_{<t'}$.
For any $t'\in [0,1]$, $0\le\eta\le |W|$, if $\eta \le \frac{|W_{< t'}| \exp((1-e_{min})\varepsilon/(2\Delta_\score))}{|W_{\ge t'}|\exp((1+\alpha)\varepsilon/(2\Delta_\score))}$ we have 
    \begin{align*}
        &\textstyle \Pr[\wdw \not\in W_{\ge t'}] \ge 
        1 - \frac{%
                1%
            }{%
                \eta + 1%
            }
        \text{,}
    \end{align*}
where $\eta$ describes the ratio between the number of $t'$-central and not $t'$-central split candidates as well as the ratio between the approximated and bounded scores. More details on this approximation is in \cref{Ssec:UtilityProof_innerQuantileSelected}.

\section{Evaluation}
\label{sec:eval}

To demonstrate that \DPM outperforms state-of-the-art clustering algorithms in terms of the standard metrics, inertia, silhouette score, and clustering accuracy while obtaining a clustering result close to that of a non-private baseline, we conduct an evaluation on synthetic and real-life data sets. First, we describe the experimental setup and provide details on the implementation of \DPM. Then, we present the results and argue that \DPM goes an important step towards being hyper-parameter free.
Additional experiments are omitted to the full version.

\subsection{Experimental Setup}

We evaluate \DPM for two synthetic and four real data sets (\Cref{tab:dts}). We generate synthetic data sets by sampling the same number of data points from multiple Gaussians, one for each cluster, for $10$ and $100$ dimensions, respectively. We obtain two real data sets by training a convolutional neural network (CNN) on the image data sets MNIST and MNIST fashion, respectively, and extract embeddings from the second-to-last layers. In addition, we use the data sets UCI Letters and UCI Gas Emissions, which consist of categorical and sensory measurements. Since the number of classes for UCI Gas Emissions is unknown, we let the non-private KMeans clustering algorithm evaluate multiple number of clusters and use those that maximises the silhouette score.

To compare \DPM with prior work, we use the following algorithms. 
The KMeans++ clustering algorithm serves as a non-private baseline and the DP-Lloyd algorithm serves as a privacy-preserving baseline. We include two other privacy-preserving clustering algorithms that form the state-of-the-art. To the best of our knowledge EM-MC~\cite{emmc21} provides the best formal guarantees and LSH-Splits~\cite{lshsplits2021} achieves the best empirical results. LSH-Splits requires the data to be centred. Therefore, we perform an additional privacy-preserving centring of the data when using LSH-Splits. We run each algorithm on each parameter configuration $20$ times and report the average metrics and the standard deviation. The corresponding range $R$ of each data set can be obtained from the meta-data of the data set or, in the case of the embeddings, by introducing a clipping layer right before the output. 

The privacy budget is $\varepsilon = 1$ and $\delta = (\tilde{n}_D \cdot \sqrt{\tilde{n}_D})^{-1}$ in all experiments.
For \DPM we distribute the privacy budget as follows: $\varepsilon = 0.04 \epsWdwSize + 0.18 \epsCount + 0.18 \epsExpMech + 0.6 \epsAverage$ and $\delta = 0.2 \delta_{\text{cnt}} + 0.8 \delta_{\text{avg}}$. 
The maximum recursion depth is $\maxdepthRec = 7$ and the minimum number of elements is $\minNumElem = \tilde{n}_D / 2^\maxdepthRec$. For the centreness, we set $t=0.3$ and $q=\frac{1}{12}$. For the other algorithms, we use the hyper-parameters from the original work.

\paragraph{Implementation Details}
To apply KMeans, we use the implementation from scikit-learn~\cite{scikit-learn} and we make use of the DiffPrivLib library~\cite{diffprivlib} for DP-Lloyd. The library~\cite{repoGoogle} provides an implementation for LSH-Splits and private averaging. The authors of EM-MC~\cite{emmc21} provide an implementation. Due to major inconsistencies to the original work, however, we re-implemented EM-MC.

\subsection{Evaluating the Clustering Quality}
\label{ssec:fixed_k}
When evaluating the clustering quality with the standard metrics, inertia, silhouette score, and clustering accuracy, a specific problem arises. The three standard metrics fail to adequately assess the quality of a clustering result due to two issues. First, inertia and clustering accuracy both profit from a growing number of clusters. However, a clustering result with many cluster centres does not imply that the clusters represent the underlying structure of the data set well. Second, the silhouette score considers distances inside a cluster and distances to the closest cluster. However, clusters that are far away from all other clusters and contain only a handful of data points have a marginal influence on the silhouette score, although such clusters do not contribute much to the clustering quality. Hence, even if privacy-preserving clustering algorithms output a clustering result of high quality in terms of the standard metrics, the clustering result can significantly differ from a non-private baseline like KMeans. 

To measure the discrepancy between the clustering of a non-private and a privacy-preserving clustering algorithm, we first conduct a non-private hyper-parameter search with the KMeans clustering algorithm to find multiple sets of clusters that maximise the silhouette score to reduce randomness. In our experiments we found that a comparison with 10 runs is already stable, however we pick the top $20$ KMeans runs and then compute the distance from each of the cluster centres in the private clustering $C_{\text{priv}}$ to the closest KMeans cluster centre for each of the stored KMeans-runs.
\begin{definition}[KMeans Distance]
\label{def:ksim}
Given $\ell$ sets of cluster centres $C_0, \dots, C_{\ell}$, obtained from the non-private KMeans clustering algorithm after a hyper-parameter search that optimises for the silhouette score and a clustering result $C_{\text{priv}}$, obtained in a privacy-preserving way, then the KMeans distance is computed as follows:
\[
\textstyle KD(C_{\text{OPT}}, C) = \frac{1}{\ell} \sum_{i = 0}^{\ell - 1} \sum_{c \in C_{\text{priv}}} \min_{c' \in C_i} ||  c - c' ||_2.
\]
with $||.||_2$ being the Euclidean norm.
\end{definition}
The KMeans distance is explicitly designed to be used in conjunction with the standard clustering quality metrics during the evaluation. We take the reported number of classes per data set from \Cref{tab:dts}. 
The results averaged over $20$ runs of each algorithm and for each data set are reported in \Cref{tab:kopt_main} and \Cref{tab:kopt_ablation}. We separate the results for the UCI data sets because both do not have well separated clusters, as implied by the silhouette score and are therefore difficult to cluster.
It can be seen that \DPM achieves state-of-the-art utility even in the case of not well separated clusters. Note that \DPM does not use the reported number of classes as input, but produces a clustering of a size close to the ground truth. In addition, the clustering result of \DPM is close to the non-private baseline in terms of the KMeans distance.

\subsection{Running Time}
\begin{figure}[t]
    \centering
    \includegraphics[width=\columnwidth]{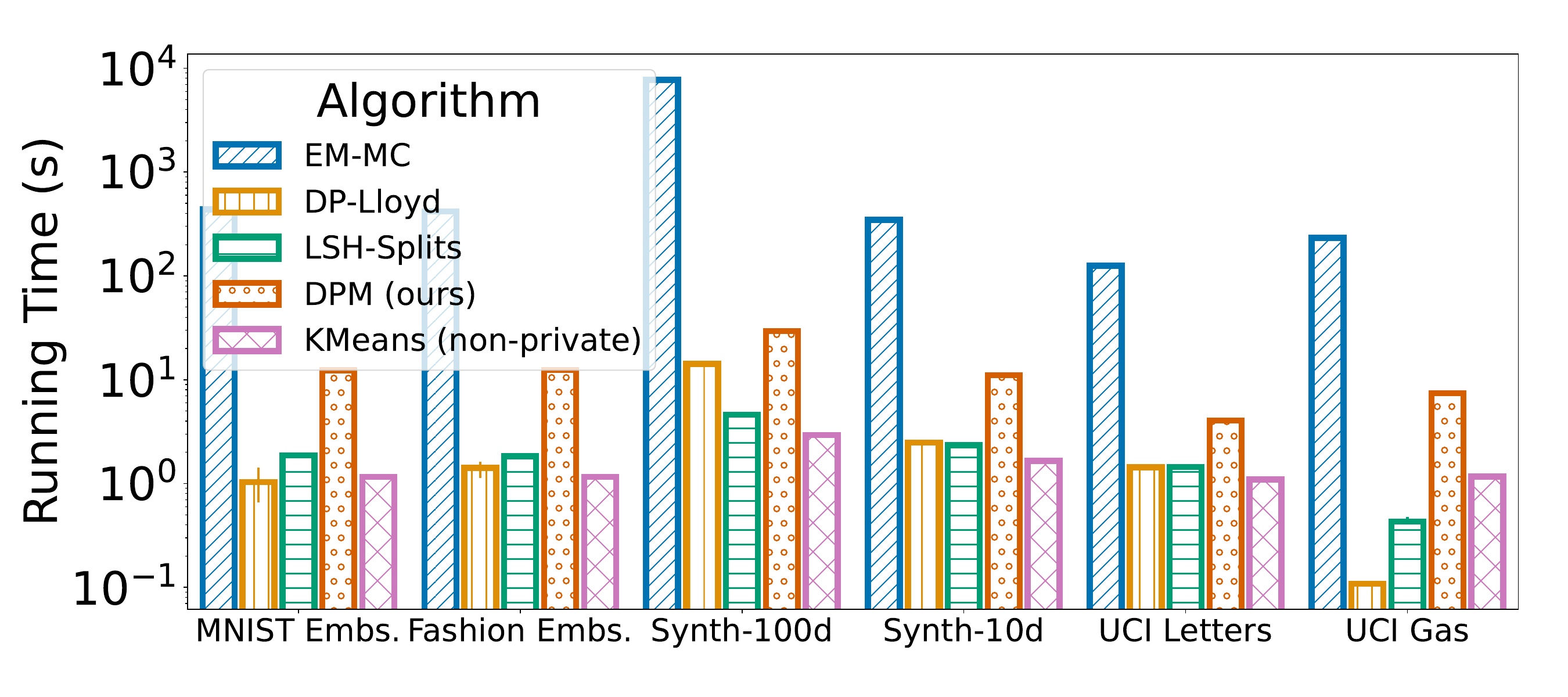}
    \caption{Evaluation of the running time (in seconds) of all algorithms including the non-private baseline averaged over $10$ runs. As shown in \cref{ssec:timeComplexityDPM}, the running time of \DPM increases significantly with a large number of dimensions and slightly with an increasing number of data points. In general, the running time of \DPM remains competitive with the other clustering algorithms.}
    \label{fig:runningTime_eval}
\end{figure}
\begin{figure}[t]
    \centering
    \includegraphics[width=\columnwidth]{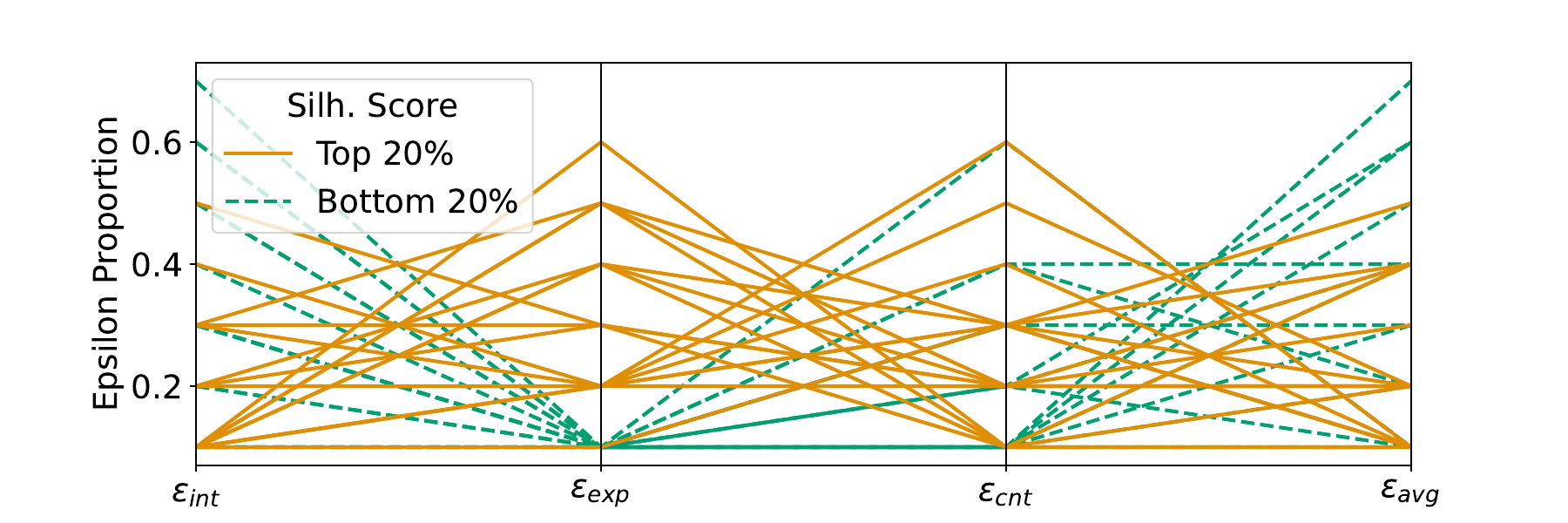}
    \includegraphics[width=\columnwidth]{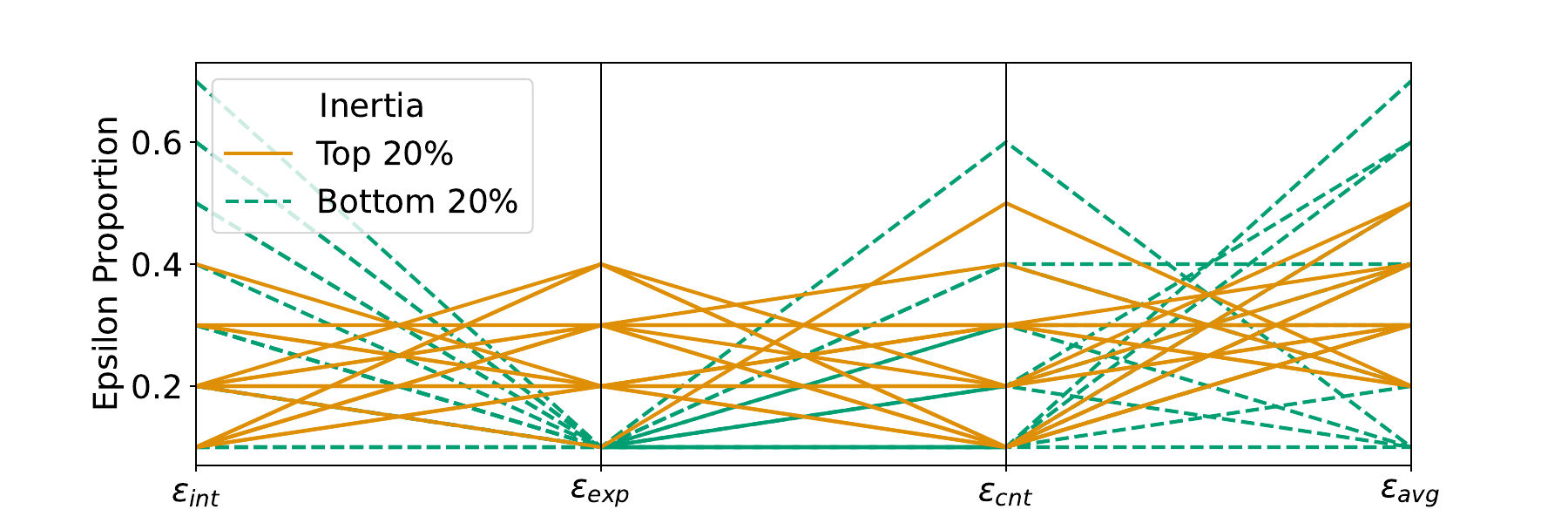}
    \includegraphics[width=\columnwidth]{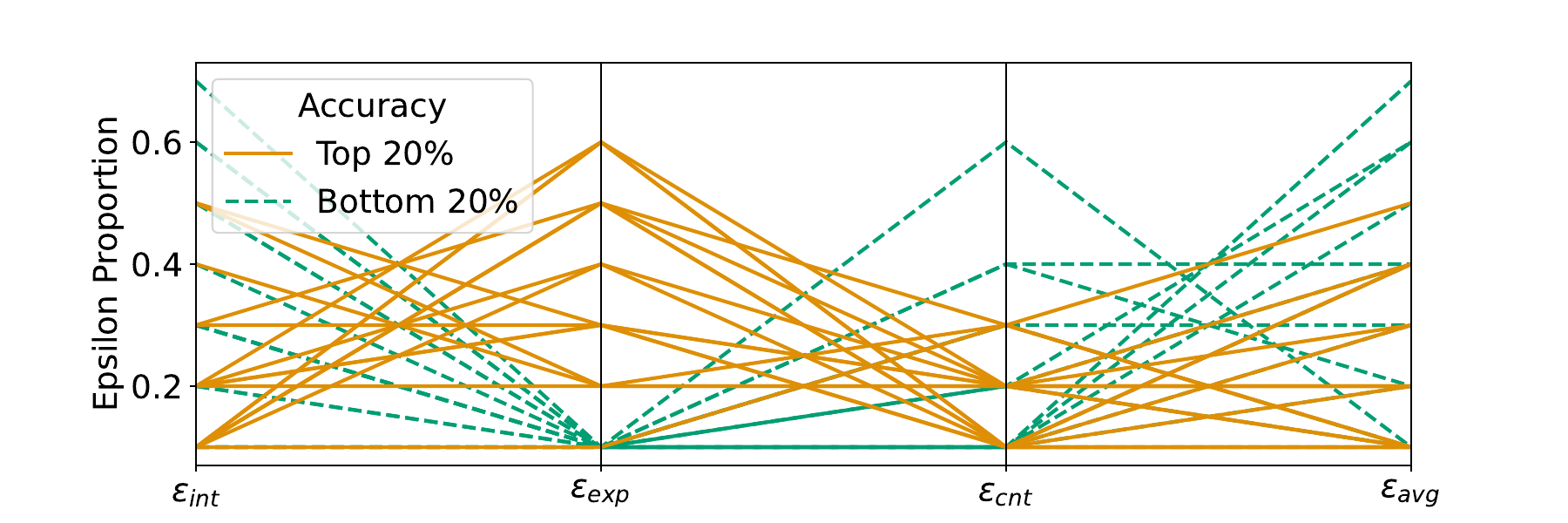}
    \includegraphics[width=\columnwidth]{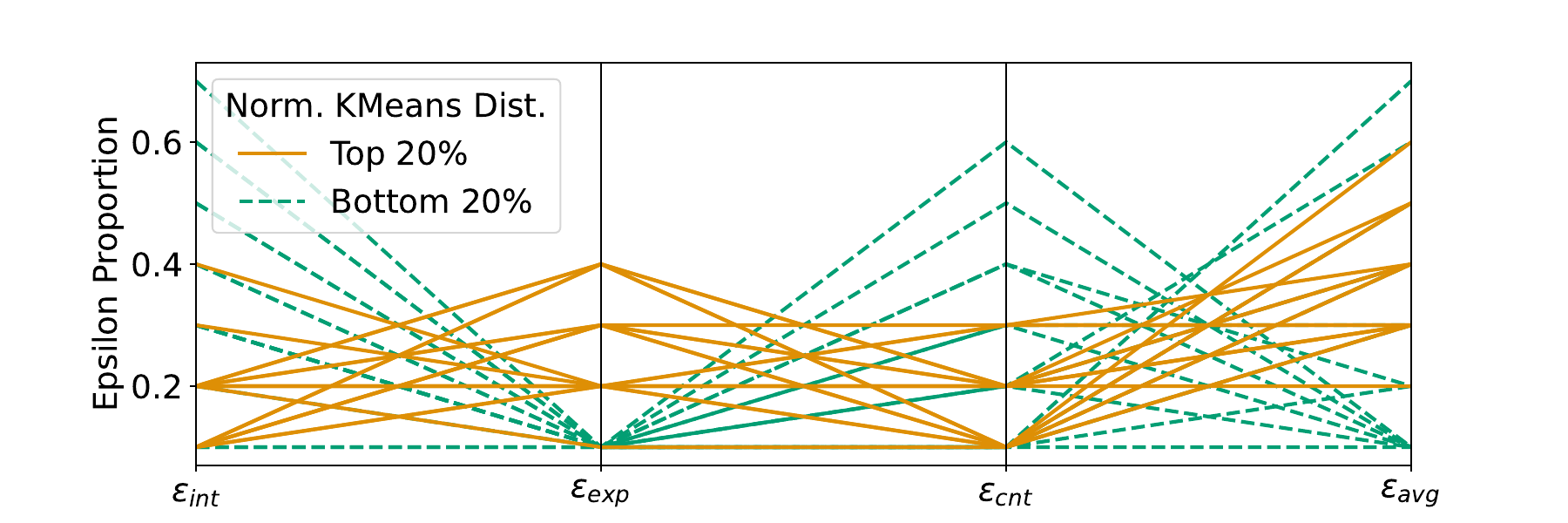}
    \caption{Evaluation of the privacy budget distribution of \DPM with respect to different clustering metrics for the MNIST Embs. data set. Each plot shows the upper (green) and lower (orange) $20\%$ of privacy budget distributions, averaged over $10$ runs with $\epsWdwSize + \epsExpMech + \epsCount + \epsAverage = 1$. Each metric prefers a slightly different privacy budget distribution, thus for the experiments we choose a distribution that performs well for all metrics and that is not specific for any data set.}
    \label{fig:epsDist_eval}
\end{figure}
\begin{figure*}[t]
    \centering
    \begin{subfigure}{0.19\textwidth}
        \includegraphics[width=\linewidth]{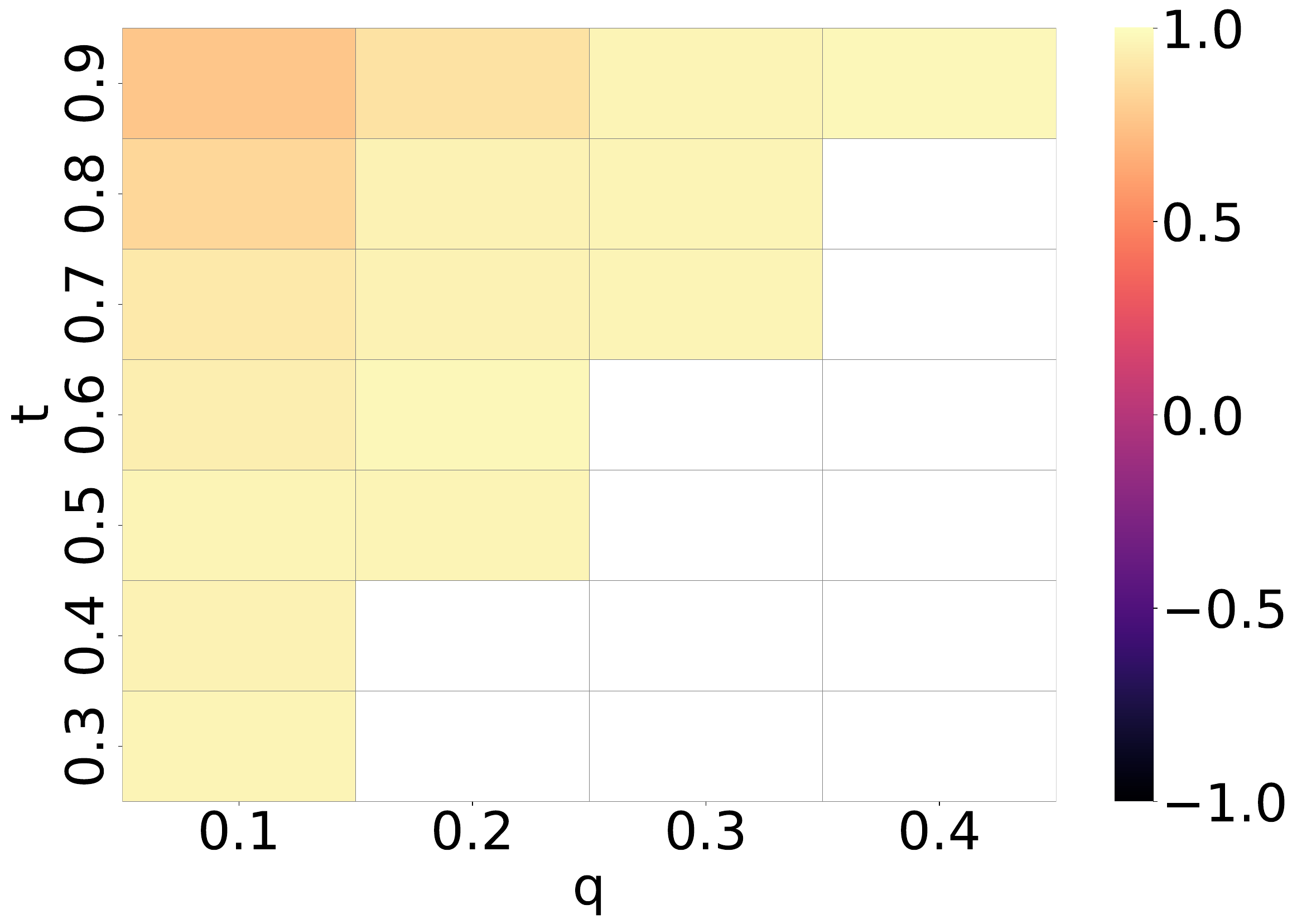}
        \caption{Silhouette Score ($\uparrow$)}
    \end{subfigure}\hfill
    \begin{subfigure}{0.19\textwidth}
        \includegraphics[width=\linewidth]{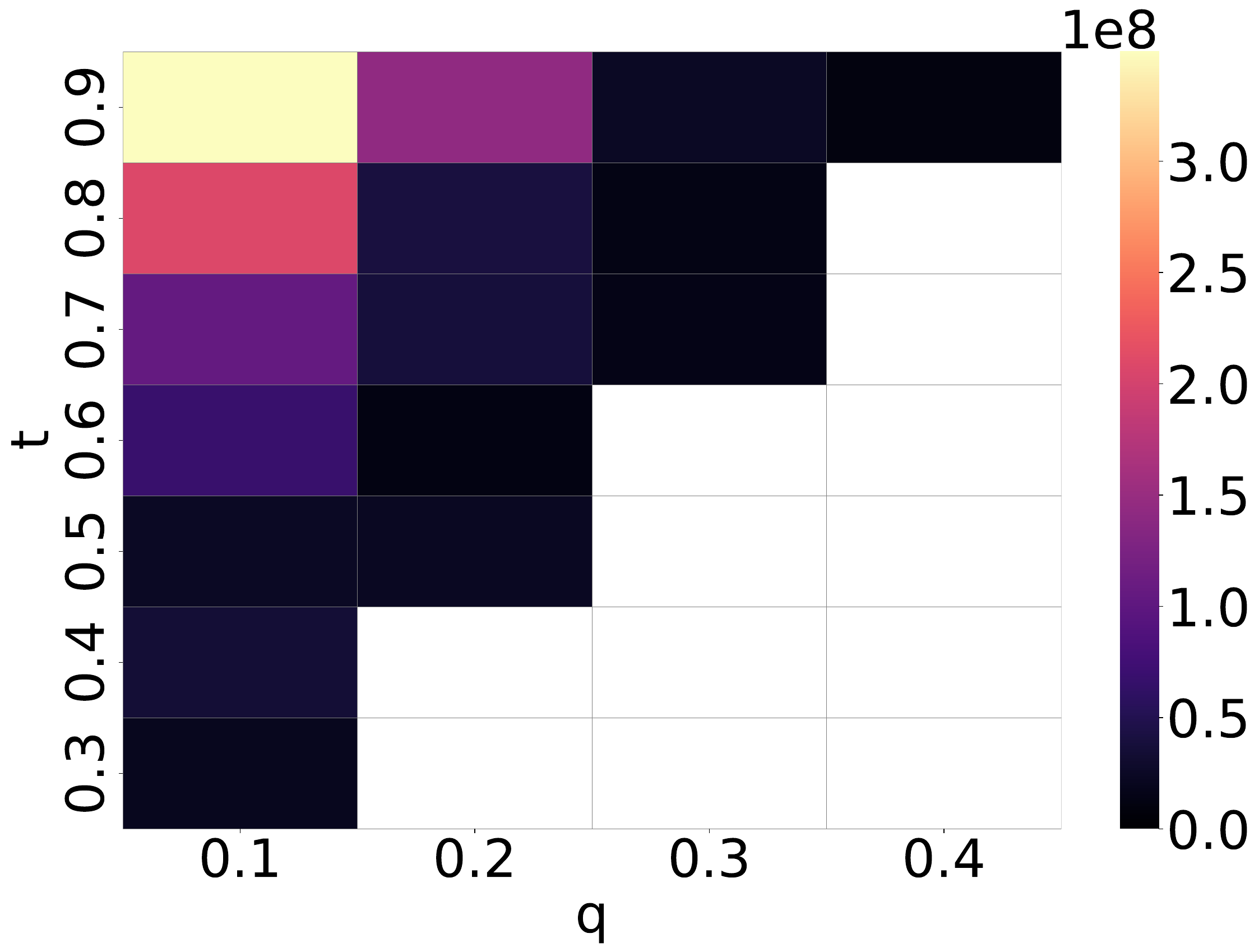}
        \caption{Inertia ($\downarrow$)}
    \end{subfigure}\hfill
    \begin{subfigure}{0.19\textwidth}
        \includegraphics[width=\linewidth]{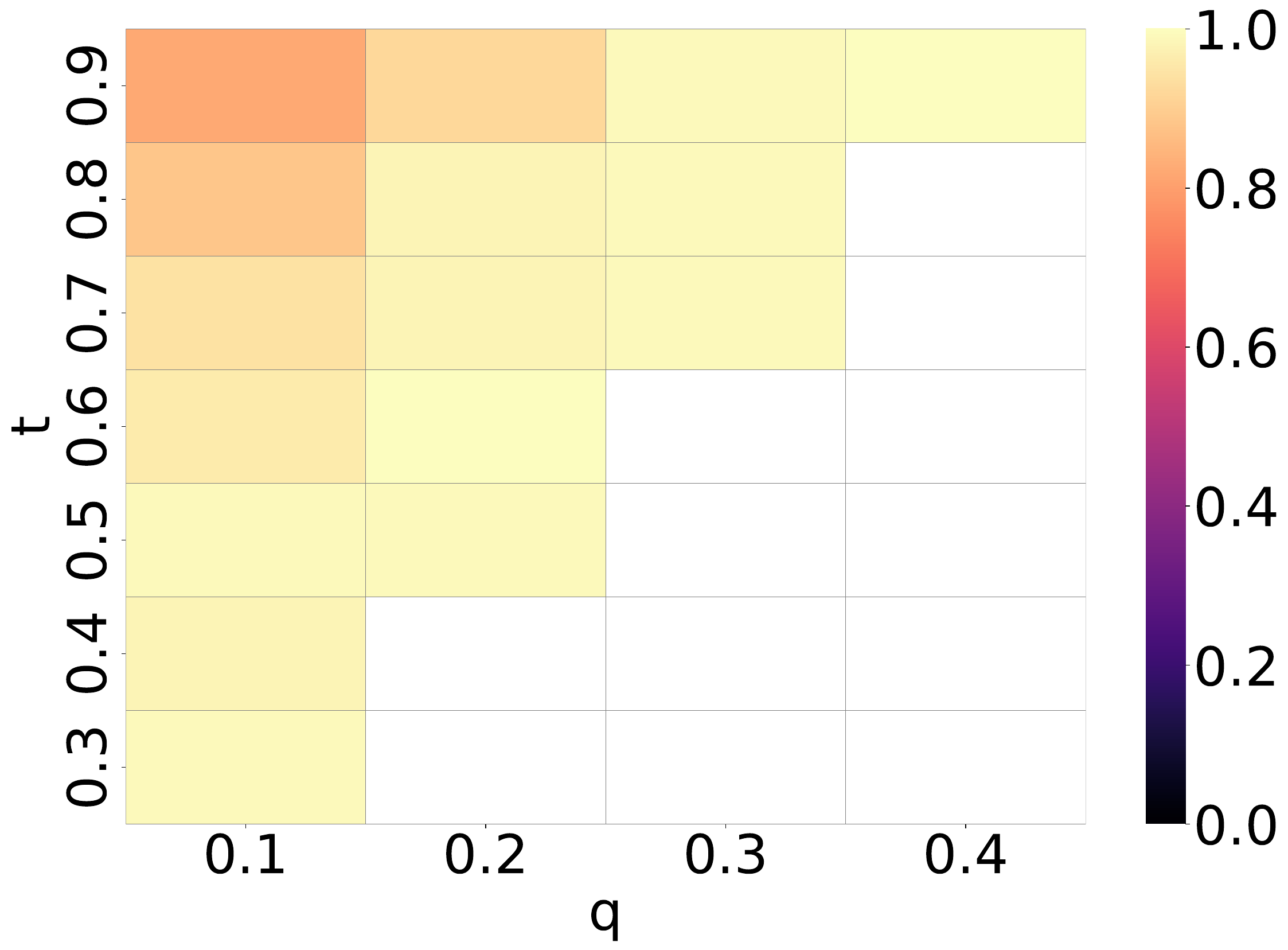}
        \caption{Accuracy ($\uparrow$)}
    \end{subfigure}\hfill
    \begin{subfigure}{0.19\textwidth}
        \includegraphics[width=\linewidth]{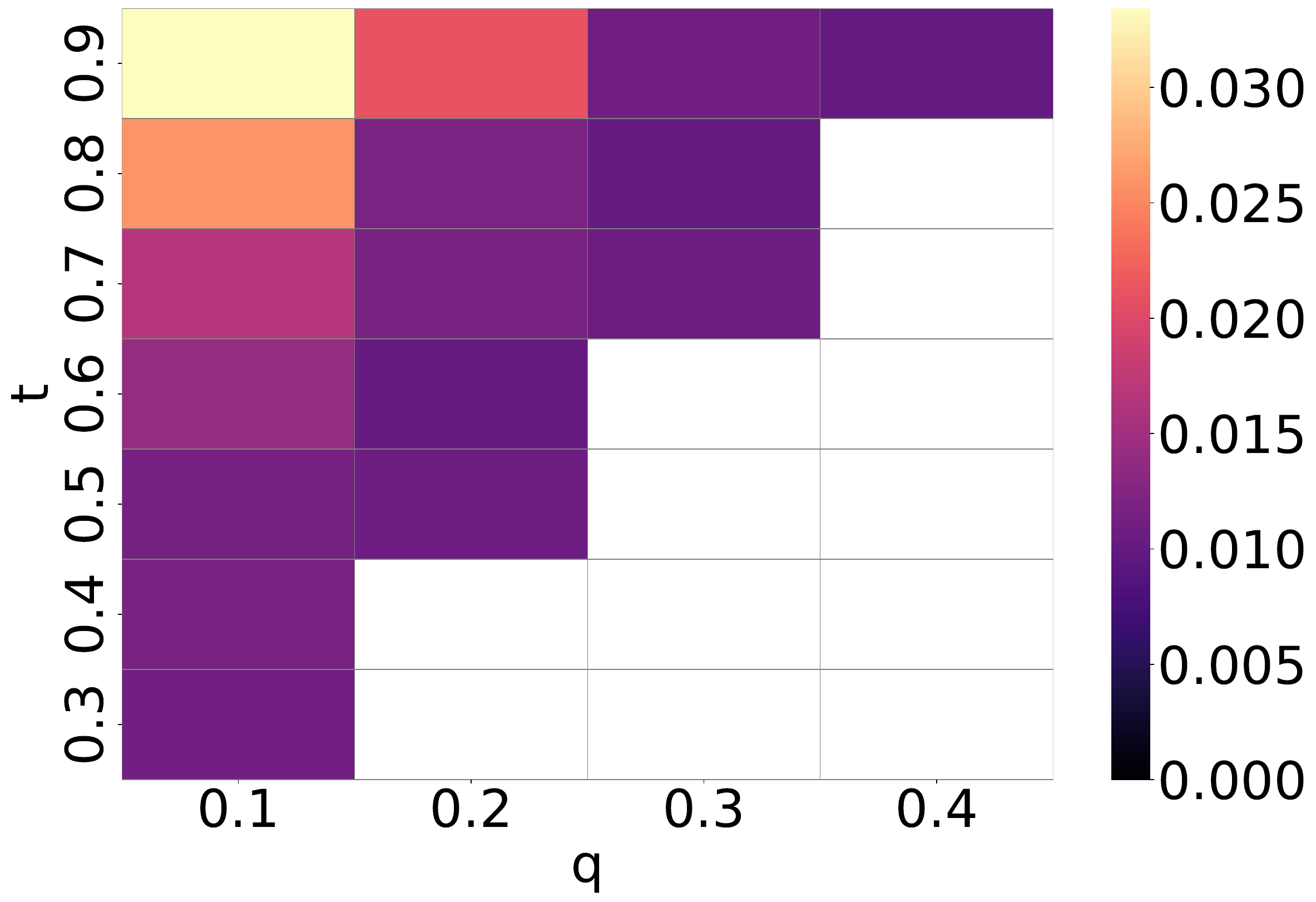}
        \caption{KMeans Dist. ($\downarrow$)}
    \end{subfigure}\hfill
    \caption{Evaluation of the effect of the centreness parameters $t, q$ on the utility of \DPM for $t \ge 2q$ evaluated using the Synth-10d dataset. Varying $t$ or $q$ does not change the clustering result of \DPM by much. On the contrary, only a specific selection reduces utility. This effect is the same for all the data sets. Thus, the centreness parameters do not require a specific hyper-parameter tuning.}
    \label{fig:cent_eval}   
\end{figure*}
To evaluate the running time of \DPM, we run all the clustering algorithms on multiple data sets and measure the elapsed time. 
We repeat the evaluation $10$ times and report the average running time combined with the standard deviation in \cref{fig:runningTime_eval}.

It can be seen that EM-MC is the slowest algorithm as the running time is in the order of minutes or even hours. 
We do not show results for EM-MC applied to Synth-100d because the running time for multiple runs exceeds our limit of $12$ hours. 
The running times for Means, DP-Lloyd and LSH-Splits are similar.
DP-Lloyd runs fast on UCI Gas. because the cluster centres converge quickly due to the difficulty of clustering and the high noise level. \DPM is slightly slower when compared to the other algorithms except for EM-MC. This is mainly due to \cref{algo:wdwsize}, but also the process of iterating all split candidates per dimension is time consuming. The influence of the number of dimensions can be seen by comparing Synth-100d and Synth-10d. Both have the same number of data points and clusters, but the difference in the number of dimensions leads to a significant increase in running time. The influence of the number of data points becomes apparent when using the UCI and Embs. data sets. Although the UCI Gas. data set has far fewer clusters and dimensions than the UCI Letters data set, the latter takes the least time because it has the smallest number of data points.    
In general, the running time of \DPM is within a few seconds even for high-dimensional data sets with a large number of data points. We would like to point out that the running time can be further improved. Firstly, \cref{algo:wdwsize} can be outsourced to an offline phase, and secondly, the highly recursive nature of \DPM allows for a high degree of parallelisation of each split level.

\subsection{Steps towards Hyper-Parameter Freeness}\label{ssec:hyperparameterfreeness}

\DPM directly incorporates the process of finding a fitting number of clusters, in contrast to other privacy-preserving clustering algorithms that have to perform a privacy-preserving hyper-parameter search for the number of clusters in advance. However, this comparison is only valid if \DPM does not require other data set-calibrated hyper-parameters. 

\DPM has the following hyper-parameters that have to be determined: emptiness-weight $\alpha$, split depth $\maxdepthRec$ and minimum number of elements in a cluster $\minNumElem$, range $R$ of the data set, and the list of sigma candidates \emph{sigmas} for the interval size estimation, the privacy budget distribution, and the centreness parameters $t, q$. 
The emptiness weight $\alpha = 5$ can be derived according to \Cref{sec:hyper_estimation}. The maximum number of splits $\maxdepthRec$ is set to $7$ to allow a maximum of $128 = 2^7$ clusters. Then, $\minNumElem$ is $\tilde{n} / 2^\maxdepthRec$. $R$ can be obtain from the meta data of the data sets without privacy leakage and we choose $\emph{sigmas} = [30]$ independent of the data set.

To show that an appropriate privacy budget distribution can be chosen without a data dependent fine-tuning, we perform an evaluation and present the results in \Cref{fig:epsDist_eval}. The figure shows the $20\%$ of the best and worst performing privacy budget distributions for each metric, exemplarily for the MNIST embs. data set, averaged over $10$ runs. 
It can be seen that one cannot directly derive the best performing configuration, but each epsilon prefers a certain range, which, in combination, helps to select a specific configuration. The privacy budget $\epsWdwSize$ should not consume too much of the total budget, $\epsAverage$ needs a bit more budget, and $\epsCount$ as well as $\epsExpMech$ should be neither too small nor too large. But it also seems that $\epsExpMech \geq \epsCount$ is preferable. In general, we conclude that although each metric has its own preference for how the privacy budget is distributed, there are some configurations that perform well across all metrics and on all data sets. In our evaluation, we use the largest portion of the privacy budget for the averaging, the smallest portion for the interval size estimation, and split the reminder equally between the Exponential and the Laplace Mechanism.

Finally, we show that the centreness-parameters can be chosen data-independently. We evaluate the effect of $t, q$ on all metrics, which is exemplary shown in \Cref{fig:cent_eval}, because the results for the other data sets are similar. Only for a large $t$ and a small $q$ utility degrades, otherwise the parameters $t,q$ nearly have the same effect on all data sets for all metrics. Thus the centreness is not a data-dependent hyper-parameter in our evaluation.

In summary, the choice of hyper-parameters for \DPM does not depend on individual data sets. Thus, unlike the other algorithms, it is not necessary to perform a hyper-parameter search for the number of clusters in advance.

\section{Related Work}
\label{sec:related_work}
Privacy-preserving clustering algorithms, in general, consider the central model of Differential Privacy and primarily focus on minimising the metric inertia~\cite{DPMaxCover21, DPClTightApprox20, DPClConstMult18, DPCl1MultCl17, DPCl1Cluster16, LocDP1Cl18, LocDPClKMeans}. However, a few algorithms~\cite{DPClStabAssump20, DPClOptDPWS18} give theoretical guarantees not in terms of inertia but in terms of the Wasserstein metric. 
Some algorithms work in the local model of Differential Privacy~\cite{LocDP1Cl18, LocDPClKMeans} that assumes a distributed setting.
Some privacy-preserving clustering algorithms work similarly to DP-Lloyd~\cite{DPClLloyd16} and try different strategies of distributing the privacy budget among each step of the algorithm~\cite{DPUtilEffCl21, DPClPBDistContour18}. Then there are some algorithm~\cite{DPMaxCover21, DPClTightApprox20, DPClConstMult18, DPCl1MultCl17} that aim to find accumulations of data points serving as candidates for cluster centres. Sometimes specific assumptions on the data set allow for more efficiency in terms of privacy~\cite{DPClEasyIns21}.
Then, there is the attempt to be more robust when it comes to higher dimensions~\cite{balkan17}. We want to highlight the algorithm EM-MC~\cite{emmc21} because, to the best of our knowledge, it currently gives the best theoretical guarantees in terms of inertia.

The approach of our algorithm \DPM is based on finding hyperplanes to subdivide the data set while keeping clusters intact. This technique is used, for example, to find heavy hitters differentially private~\cite{PrivTree} via data-independent splits and a scoring function to decide whether further splitting is needed. The approach of searching for separating hyperplanes can also be found in supervised learning methods such as SVMs~\cite{ConvexOpt, RiskMinimization} and decision trees~\cite{Maddock}.   

Next, we give a more detailed overview of clustering algorithms that are close to \DPM. 
We start with Optigrid~\cite{Optigrid99}, which is a non-private clustering algorithm. To counter the curse of dimensionality, Optigrid considers each dimension of the data separately and uses a dimension-wise kernel density estimation to find a split through a sparse region. Since Optigrid is a non-private algorithm that uses techniques like kernel density estimation, it cannot directly be transferred to Differential Privacy. 

LSH-Splits~\cite{lshsplits2021} is a privacy-preserving clustering algorithm that reaches state-of-the-art utility results when evaluated on inertia, silhouette score, and clustering accuracy. It subdivides the data set into small subsets by sampling splits randomly through the mean of the data points. 
The resulting subsets yield high utility when synthesising the intended number of clusters via a non-private clustering algorithm. Since splits are bound to the mean of the data set, LSH-Splits is naturally limited in the variety of possible splits.

Mondrian~\cite{mondrian06} is not a clustering algorithm but rather aims to subdivide the data set into equally sized parts. Thereby, it recursively splits the data set by considering each dimension separately and searching for the dimension of the highest variance and splits along the median of that dimension. A similar approach is adopted by our proposed algorithm \DPM. However, \DPM does not only take into consideration the distance to the median but also the number of data points in the vicinity of a split candidate.

\section{Conclusion}
\label{sec:conclusion}
In this work, we propose a differentially private clustering algorithm called \DPM that achieves state-of-the-art performance on the standard clustering metrics inertia, silhouette score and accuracy. The clustering produced by \DPM is also much closer to the clustering produced by the non-private KMeans algorithm. 

\DPM looks for separations of the data set called splits that ultimately preserve clusters of data points. In the process, \DPM generates multiple split candidates for each dimension of the data. All split candidates are scored based on a scoring function that considers the relative position of the split candidate to the median of the data along a dimension (centreness) and the number of data points in the vicinity of a split candidate (emptiness). Finally, \DPM has a high probability of selecting a split with a high score.

\DPM makes an important step towards hyper-parameter freeness. We present a DP algorithm for estimating the granularity of the data, i.e., an approximation of the distance between clusters of data points. In addition, we provide an analytical discussion of balancing various optimisation terms, and show experimentally for other parameters that \DPM's utility performance is insensitive to their choice. Furthermore, \DPM stops if no reasonable separations are found. Thus, no hyper-parameter search for the optimal number of clusters $k$ needs to be performed.

We prove DP guarantees for \DPM and characterise its behaviour under more general assumptions on the input data set. We also provide an extensive empirical evaluation of the clustering quality of \DPM, not only in terms of the standard metrics but also in terms of a novel metric, the KMeans distance.

\paragraph{Acknowledgements.}
The research received funding from the BMBF in the project MLens and from the BMBF and EU in the project AnoMed.

\bibliographystyle{abbrv}
\bibliography{ccs24/sources}

\appendix

\section{Postponed Privacy Proofs}

\subsection{Proof of \cref{lem:PrivacyPercentile}}
\label{Ssec:PrivacyPercentile}
\begin{LemmaRepetition}{\ref{lem:PrivacyPercentile}}
    Algorithm DP-PERCENTILE ($4.3$,~\cite{expQ2020}) applied to distances between neighbouring data points has a sensitivity of $2$. 
\end{LemmaRepetition}
\begin{proof}
    The original proof in \cite{expQ2020} considers neighbouring data sets and exchanging a single element can only influence the rank of an element (the distance between the index of that element and the index of the median) by one. We have a similar setting but instead of data points, we consider distances between neighbouring data points. Then, exchanging a single data point can change two distances in the input, thus the rank of an element in the input can change by a maximum of two.
\end{proof}

\subsection{Proof of \cref{lem:sensitivity_centreness}}
\label{Ssec:PrivacyProofCentreness}
\begin{LemmaRepetition}{\ref{lem:sensitivity_centreness}}
    Let $S, S'\in \datasets$ be neighbouring data sets and noisy count $\tilde n$ and offset $\lambda \ge -\ln(2\deltaCount)/\epsCount$ with $\tilde n -\lambda \le |S|,|S'|$. 
   Then, the sensitivity for the centreness with the parameters $t, q$ defined as in \cref{def:cent}, is $\Delta_{\cent_{t,q}} = \frac{t}{(\tilde{n}-\lambda)q}$.
\end{LemmaRepetition}
\begin{proof}
    \label{Sssec:SensGenCent}
    Let $S, S'\in \datasets$ be neighbouring data sets and the corresponding noisy counts $\tilde n$ and offset $\lambda \ge -\ln(2\deltaCount)/\epsCount$ with $\tilde n -\lambda \le |S|,|S'|$. W.l.o.g. we assume $|S'| = |S|+1$.
    By using the shifted noisy count instead of the exact number of elements and $\frac{1}{|S|},\frac{1}{|S'|} \le \frac{1}{\tilde n -\lambda}$, we get an upper bound for the sensitivity of the centreness function.
    Since the centreness consists of two linear functions, the sensitivity is determined by the one with maximum slope. 
    As the centreness function is symmetric, we consider the ranks $\rank$, $\rank'$ in the interval $[0,\frac{\tilde{n}}{2}]$.
    As we require $t\ge 2q$, the maximum slope for any rank $\rank\in [0,\tilde n]$ is in the outer quantile.
    \begin{align*}
        &~| \centFull{S}{\tilde n}{\wdw} -\centFull{S'}{\tilde n }{\wdw} | \\
        =&~\bigg| \frac{(\frac{\tilde n}{2}- |\rank-\frac{\tilde n}{2}|)\cdot t)}{\tilde n q} -\frac{(\frac{\tilde n}{2}- |\rank'-\frac{\tilde n}{2}|)\cdot t)}{\tilde n q} \bigg| \\
        = &~\bigg| \frac{(\frac{\tilde n}{2}- |\rank-\frac{\tilde n}{2}|)\cdot t)-(\frac{\tilde n}{2}- |\rank'-\frac{\tilde n}{2}|)\cdot t}{\tilde n q} \bigg| \\
        = &~ \frac{t}{\tilde nq}\bigg|-\big|\rank-\frac{\tilde n}{2}\big|+\big|\rank'-\frac{\tilde n}{2}\big|\bigg| \\
        \le  &~ \frac{t}{(\tilde n -\lambda)q}\bigg|-\big|\rank-\frac{\tilde n}{2}\big|+\big|\rank'-\frac{\tilde n}{2}\big|\bigg| \\
    \intertext{The rank of $s$ can either be the same for $S$ and $S'$ or it can differ by one.
    If the ranks are not the same ($\rank' = \rank+1$), we get the following equation:}
        = &~ \frac{t}{(\tilde n -\lambda)q}\bigg|-\big|\rank-\frac{\tilde n}{2}\big|+\big|\rank+1-\frac{\tilde n}{2}\big|\bigg| 
    \intertext{Now because we consider the absolute value, we have to consider the following three cases for $0 \le \rank\le \tilde n$.
    First, we consider the case that $\rank < \tilde n/2 $, than we also know that $\rank < \tilde n/2 -1$ and get}
        = &~ \frac{t}{(\tilde n -\lambda)q}\bigg|\underbrace{\rank-\frac{\tilde n}{2}-\rank-1+\frac{\tilde n}{2}}_{=-1}\bigg|=~ \frac{t}{(\tilde n -\lambda)q} 
    \intertext{Next, we consider the case that $\rank > \tilde n/2 - 1$ which also implies that $\rank > \tilde n/2 $ which gives us}
        = &~ \frac{t}{(\tilde n -\lambda)q}\bigg|\underbrace{-\rank+\frac{\tilde n}{2}+\rank+1-\frac{\tilde n}{2}}_{=1}\bigg| = ~ \frac{t}{(\tilde n -\lambda)q}\\
    \intertext{The last case that we have to consider is that the first term needs no be negated but not the second: $ \tilde n/2 -1 <\rank <  \tilde n/2 $.}
        = &~ \frac{t}{(\tilde n -\lambda)q}\bigg|\underbrace{\rank-\frac{\tilde n}{2}+\rank+1-\frac{\tilde n}{2}}_{=2\rank-\tilde n}\bigg| = ~ \frac{t}{(\tilde n -\lambda)q}\big|\underbrace{2\rank-\tilde n+1}_{\le 1}\big| \\
    \intertext{Otherwise, if the rank of $s$ is the same for $S$ and $S'$ ($\rank'=\rank$), we get:}
            &~| \centFull{S}{\tilde n}{\wdw} -\centFull{S'}{\tilde n }{\wdw} | \\
            =&~\bigg| \underbrace{\frac{(\frac{\tilde n}{2}- |\rank-\frac{\tilde n}{2}|)\cdot t)}{\tilde n q} -\frac{(\frac{\tilde n}{2}- |\rank-\frac{\tilde n}{2}|)\cdot t)}{\tilde n q} }_{=0}\bigg| = ~ 0\\
    \end{align*}
\end{proof}

\subsection{Proof of \cref{lem:sensitivity_emptiness}}
\label{Ssec:PrivacyProofEmptiness}
\begin{LemmaRepetition}{\ref{lem:sensitivity_emptiness}}
    Let $S, S'\in \datasets$ be neighbouring data sets, noisy count $\tilde n$ and offset $\lambda \ge -\ln(2\deltaCount)/\epsCount$ with $\tilde n -\lambda \le |S|,|S'|$. 
    Then, the sensitivity of the subscore emptiness for the split interval size $\wdwSize$ defined as in \Cref{def:empt_frac} is $\Delta_\empt = \frac{1}{\tilde{n}-\lambda}$.
\end{LemmaRepetition}
\begin{proof}
    The emptiness of a split candidate $\wdw$ is computed as in \Cref{def:empt_frac}.  Let $S, S'\in \datasets$ be neighbouring data sets, noisy count $\tilde n$ and offset $\lambda \ge -\ln(2\deltaCount)/\epsCount$ with $\tilde n -\lambda \le |S|,|S'|$. By using the noisy count instead of the exact number of elements and $\frac{1}{|S|}, \frac{1}{|S'|} < \frac{1}{\tilde n-\lambda}$, we get an upper bound for the sensitivity of the emptiness function. We distinguish the cases that the additional element is in the split interval of $\wdw$ for $S'$ and that it is not in the split interval of $\wdw$.  W.l.o.g. we assume that $|S'| = |S|+1$ and $\wdwxi = |\wdw|$ as the number of elements in $s$ for $S$ and $\wdwxi' = |\wdw|$ the number of elements in $s$ for $S'$.
    Furthermore, we assume that $\wdwxi$ is the number of elements in $\wdw$ for $S$ and $\wdwxi' = \wdwxi +1 = |\wdw|$ for $S'$.  We consider the case that the additional element in $S'$ is in $s$ and the case that it is in any other split candidate. In the first case, we know $\wdwxi' = \wdwxi +1 $ and  
        \begin{align*}
            &~| \emptFull{S}{\tilde n}{\wdw} -\emptFull{S'}{\tilde n}{\wdw} | \\
            =&~ \bigg| (1 - \frac{\wdwxi}{\tilde{n}}) - (1 - \frac{\wdwxi +1}{\tilde{n}}) \bigg| \\
            =&~ \bigg| -\frac{\wdwxi}{\tilde{n}} +\frac{\wdwxi + 1}{\tilde{n}} \bigg| =~ \frac{1}{\tilde{n}} \le \frac{1}{(\tilde n -\lambda)}\\
        \intertext{For the latter case that the additional element is not in $s$, we know that $\wdwxi'=\wdwxi$ which gives us the following bound for the sensitivity.}
            &~| \emptFull{S}{\tilde n}{\wdw} -\emptFull{S'}{\tilde n}{\wdw} |\\
            =&~\bigg| \big(1 - \frac{\wdwxi}{\tilde{n}}\big) - \big(1 - \frac{\wdwxi}{\tilde{n}}\big) \bigg|= 0
        \end{align*}
\end{proof}

\subsection{Proof of \cref{lem:sensitivity_scoring}}
\label{ssec:PrivacyProofScore}
\begin{LemmaRepetition}{\ref{lem:sensitivity_scoring}}
   Let $S, S'\subseteq D \in \datasets$ be neighbouring data sets, noisy count $\tilde n$ and offset $\lambda \ge -\ln(2\deltaCount)/\epsCount$ with $\tilde n -\lambda \le |S|,|S'|$. 
   With the sensitivity for the subscores centreness and emptiness, and the weight $\alpha$, the sensitivity $\Delta_\score$ for the score of a split candidate in $S$ is given by $\textstyle \Delta_\score \le \frac{\frac{t}{q} + \alpha }{\tilde{n}-\lambda}$.
\end{LemmaRepetition}
\begin{proof}
    The scoring function is the weighted sum of the centreness and the weighted emptiness of a split. 
    Let $S, S'\in \datasets$ be neighbouring data sets, the corresponding noisy counts $\tilde n$ and offset $\lambda \ge -\ln(2\deltaCount)/\epsCount$ with $\tilde n -\lambda \le |S|,|S'|$.
    By using the noisy count instead of the exact number of elements and $\frac{1}{|S|}, \frac{1}{|S'|} \le \frac{1}{\tilde n - \lambda}$, we get an upper bound for the sensitivity of the scoring function. We use \cref{lem:sensitivity_emptiness} and \cref{lem:sensitivity_centreness} to bound the difference of the score of a split candidate for neighbouring data sets.
    \begin{align*}
        & ~~| \scoreF{S}{\tilde n}{\wdw} - \scoreF{S'}{\tilde n}{\wdw} |\\
        =& ~~| \alpha\emptF{S}{\tilde n}{\wdw} +\centF{S}{\tilde n_{S}}{\wdw} - (\alpha\emptF{S'}{\tilde n }{\wdw} +\centF{S'}{\tilde n}{\wdw}) |\\
        =& ~~| \alpha \underbrace{(\emptF{S}{\tilde n}{\wdw} -\emptF{S'}{\tilde n}{\wdw})}_{\le\frac{1}{\tilde{n}-\lambda}}+\underbrace{(\centF{S}{\tilde n_{S}}{\wdw}
        -\centF{S'}{\tilde n}{\wdw})}_{\le\frac{t}{(\tilde{n} -\lambda)q}} |\\
        \le& ~~ \frac{\alpha }{\tilde{n}-\lambda} + \frac{t}{(\tilde{n} -\lambda)q} = \frac{\frac{t}{q} + \alpha }{\tilde{n}-\lambda}
    \end{align*}
\end{proof}

\section{Postponed Utility Proofs}

\subsection{Proof of \cref{lem:utilityNoisyCount}}
\label{Ssec:UtilityProofNoisyCount}
\begin{LemmaRepetition}{\ref{lem:utilityNoisyCount}}
    Let $S\subseteq D$ be the considered set and $|S|$ the number of elements in this set and $\tilde n - \lambda = |S| + \Lap(1/\epsCount) - \lambda$. Then, for any $\kappa \ge 0$ we have 
    \[\Pr[\big|\tilde n -\lambda - |S|\big| > \kappa] \le 1/2\exp(-\kappa \cdot \epsCount) / \delta  \text{.}\]
\end{LemmaRepetition}

\begin{proof}
We first bound the difference between the noisy shifted count and the real count shifted by the offset $\lambda$. Then, with the estimate for $\lambda$ we get lower bounds for the difference between the noisy shifted count and the real count.

First, for some $\kappa' \ge 0$, we can bound the difference between $\tilde n'=|S| + \Lap(1/\epsCount) - \lambda$ and $|S| - \lambda$ as follows:
\begin{align}
    & \Pr[|\tilde{n}' - (|S| - \lambda)| \le \kappa'] \nonumber\\
    & = \Pr[|\Lap(1/\epsCount)| \le \kappa'] \nonumber\\
    &\ge cdf_{\Lap(1/\epsCount)}(\kappa')- cdf_{\Lap(1/\epsCount)}(-\kappa')\nonumber\\
    &=1 - \exp(-\kappa'\epsCount)\nonumber
\end{align}

If we combine this bound with the estimate for $\lambda$ from \Cref{eq:offsetEstimation}, we get for $\kappa = \kappa' + \lambda$
\begin{align}
    &\Pr[|\tilde n' - |S| + \lambda| > \kappa'] \le \exp(-(\kappa - \lambda) \cdot \epsCount)\nonumber\\
    \intertext{By the triangle inequality, $|a| - |b| \le |a + b|$. With $\tilde n' - |S| = a$ and $\lambda = b$, we know that $|\tilde n' - |S|| - |\lambda| \le |\tilde n' - |S| +  \lambda|$. As the $\Pr[X - \tau > c] \le \Pr[X > c]$ for positive $\tau$, we conclude}
    \implies
    &\Pr[|\tilde n' - |S|| - |\lambda| > \kappa'] \le \exp(-(\kappa - \lambda) \cdot \epsCount)\nonumber\\
    \Leftrightarrow
    &\Pr[|\tilde n' - |S|| > \kappa - \lambda + |\lambda|] \le \exp(-(\kappa - \lambda) \cdot \epsCount)\nonumber\\
    \intertext{By \Cref{eq:offsetEstimation} $\lambda \ge -\ln(2\deltaCount)/\epsCount$}
    \implies&
    \Pr[|\tilde n' - |S|| > \kappa] \le \exp(-\kappa \cdot\epsCount -  \ln(2\deltaCount))\nonumber\\
    \Leftrightarrow&\Pr[|\tilde n' - |S|| > \kappa] \le 1/2\exp(-\kappa \cdot \epsCount) / \deltaCount    \label{eq:closenessNoisyCount}
\end{align}
\end{proof}

\subsection{Proof of \cref{thm:exponential_mechanism_utility}}
\label{ssec:utility_generalisedExpMech}
\begin{TheoremRepetition}{\ref{thm:exponential_mechanism_utility}}
    Fixing a set $S\subseteq D \in \datasets$ and the set of candidates $W$, let $\omega \ge 0$, $OPT(S,\score,W) = \max_{\wdw\in W} \scoreF{S}{\tilde n}{\wdw}\}$ and $W_{OPT_\omega} = \{\wdw \mid |\scoreF{S}{\tilde n}{\wdw} - OPT(S,\score,W)| \le \omega\}$ denote the set of elements in $W$ which up to $\omega$ attain the highest score \break $OPT(S,\score,W)$. Then, for some $\kappa>0$ the Exponential Mechanism $M_E$ satisfies the following property:
\begin{align*}
    &\textstyle\Pr\big[\scoreF{S}{\tilde n}{M_{E}(S,\score,\varepsilon)} \le OPT(S,\score,W) - \omega \\
    &\textstyle \phantom{\Pr\big[} - \frac{2\Delta_{\score}}{\varepsilon}\left(\ln\left( \frac{|W|}{|W_{OPT_\omega}|} \right)+ \kappa \right)\big] \le e^{-\kappa}\text{.}
\end{align*}
\end{TheoremRepetition}
\begin{proof}
    We adapt the proof from ~{\cite[Theorem 3.11]{DwoRo14}} to the case that optimal candidates can differ by at most $\omega$ from the optimal score.
    \begin{align*}
        &\Pr[\score(S,n,M_E(S,\score,\varepsilon))\le c]
        \intertext{We know that for all $\wdw\in W$ with $\scoreF{S}{n}{s}\le c$ the unnormalized probability mass is at most $\exp(\varepsilon c/2\Delta_\score)$. We know that there are at least $|W_{OPT_\omega}|$ elements with score at least $OPT(S,\score, W) -\omega$.}
        &\le \frac{|W|\exp(\varepsilon c /2\Delta_\score )}{|W_{OPT_\omega}|\exp(\varepsilon(OPT(S,\score, W)-\omega)/2\Delta_\score)}\\
        &= \frac{|W|}{|W_{OPT_\omega}|}\exp\bigg({\frac{\varepsilon(c-OPT(S,\score,W)+\omega)}{2\Delta_\score}}\bigg)
    \end{align*}
The theorem follows with 
\[
    c=OPT(S,\score,W) - \omega - \frac{2\Delta_{\score}}{\varepsilon}\left(\ln\left( \frac{|W|}{|W_{OPT_\omega}|} \right)+ \kappa \right)\text{.}
\]
\end{proof}

\subsection{Proof of \cref{lem:utility_inner_quantile}}
\label{Ssec:UtilityProof_innerQuantile}
\begin{LemmaRepetition}{\ref{lem:utility_inner_quantile}}
Let $S\subseteq D$ be a set with $\tilde n$ as the noisy number of  elements in $S$.
Let $\wdw^*$ be the split candidate with the optimal loss $\scoreF{S}{\tilde n}{\wdw^*} = OPT(S,\score, W) = \max_{\wdw \in W} \scoreF{S}{\tilde n}{\wdw}$ and $W_{OPT_\omega} = \{\wdw \mid |\scoreF{S}{\tilde n}{\wdw} - OPT(S,\score,W)| \le \omega\}$. For any $t' \in [0,1]$, if 
$\wdw^*$ with $\wdw^* \in W_{\ge t'}$, then
    \begin{align*}
    &\textstyle\Pr\big[ \left(\frac{\frac{2t}{q \alpha} + 2 }{(\tilde n - \lambda)\varepsilon} \right) \left(\ln\left( |W|/|W_{OPT_\omega}|\right)+ \kappa + \omega \right) + \frac{1-t'}{\alpha}\\
    &\phantom{\textstyle\Pr\big[ } > \emptF{S}{\tilde n}{\wdw^*} -\emptF{S}{\tilde n}{\wdw} \mid \wdw \in W_{\ge t'}\big] \ge 1 - e^{-\kappa} \text{.}
\end{align*}
\end{LemmaRepetition}
\begin{proof}
\label{proof:central_split}
To compute the probability of the chosen split \break $\wdw = M_{E}(S,\score,\varepsilon)$ to have high emptiness $\emptF{S}{\tilde n}{\wdw}$ under the assumption that it is an inner split $\wdw \in W_{\ge t'}$, we apply \Cref{cor:dpmondrian_utility_exp} and simplify the bounds. 

\begin{align*}
    &\textstyle\Pr\big[\scoreF{S}{\tilde n}{M_{E}(S,\score,\varepsilon)} \le~~OPT(S,\score,W) - \omega\\
    &\textstyle  ~-\left(\frac{2t}{(\tilde n - \lambda)\varepsilon q} + \frac{2\alpha}{(\tilde n - \lambda)\varepsilon} \right) \left(\ln\left( \frac{|W|}{|W_{OPT}|}\right)+ \kappa  \right)\big] 
    \le e^{-\kappa}\\
    \intertext{We have to refactor the inequality as we are interested in the counter probability. The probability that the difference between our score and the optimal is smaller than some value, is larger than a given probability.}
    \Leftrightarrow & \textstyle\Pr\big[\scoreF{S}{\tilde n}{M_{E}(S,\score,\varepsilon)} >~~OPT(S,\score,W)- \omega\\
    &\textstyle  ~-\left(\frac{\frac{2t}{q } + 2\alpha }{(\tilde n - \lambda)\varepsilon} \right) \left(\ln\left( \frac{|W|}{|W_{OPT_\omega}|}\right)+ \kappa  \right)\big] 
    \ge 1 - e^{-\kappa}  
    \intertext{
    We set 
    \[OPT(S,\score, W) = \centF{S}{\tilde n}{\wdw^*} + \alpha \emptF{S}{\tilde n}{\wdw^*}\] 
    and 
    \[\scoreF{S}{\tilde n}{M_{E}(S,\score,\varepsilon)} = \centF{S}{\tilde n}{\wdw} + \alpha \emptF{S}{\tilde n}{\wdw}\]
    with $\wdw = M_{E}(S,\score,\varepsilon)$.}
    \Leftrightarrow &\textstyle\Pr\big[\centF{S}{\tilde n}{\wdw} + \alpha \emptF{S}{\tilde n}{\wdw}>~~\centF{S}{\tilde n}{\wdw^*} + \alpha \emptF{S}{\tilde n}{\wdw^*}- \omega\\
    &\textstyle  ~-\left(\frac{\frac{2t}{q } + 2\alpha }{(\tilde n - \lambda)\varepsilon}  \right) \left(\ln\left( |W|/|W_{OPT_\omega}|\right)+ \kappa \right)\big] \\
    &
    \ge 1- e^{-\kappa}\\
    \Leftrightarrow &\textstyle\Pr\big[ \left(\frac{\frac{2t}{q } + 2\alpha }{(\tilde n - \lambda)\varepsilon} \right) \left(\ln\left( |W|/|W_{OPT_\omega}|\right)+ \kappa  \right) ~~ \\
    &\textstyle  ~> \centF{S}{\tilde n}{\wdw^*} + \alpha \emptF{S}{\tilde n}{\wdw^*} - \centF{S}{\tilde n}{\wdw} - \alpha \emptF{S}{\tilde n}{\wdw}- \omega\big] \\
    &\ge 1 - e^{-\kappa}\\
    \Leftrightarrow 
    &\textstyle\Pr\big[ \left(\frac{\frac{2t}{q } + 2\alpha }{(\tilde n - \lambda)\varepsilon} \right) \left(\ln\left( |W|/|W_{OPT_\omega}|\right)+ \kappa  \right)+ \centF{S}{\tilde n}{\wdw}\\
    &\textstyle  \phantom{\Pr\big[ }  - \centF{S}{\tilde n}{\wdw^*}~~ 
       > \alpha  \emptF{S}{\tilde n}{\wdw^*} - \alpha \emptF{S}{\tilde n}{\wdw}- \omega\big] 
    \\
    &\ge 1 - e^{-\kappa}\\
    \intertext{We set $g(\kappa):=\left(\frac{\frac{2t}{q \alpha} + 2 }{(\tilde n - \lambda)\varepsilon} \right) \left(\ln\left( |W|/|W_{OPT_\omega}|\right)+ \kappa  \right)$}
    \Leftrightarrow &\textstyle\Pr\big[ g(\kappa) + \frac{\centF{S}{\tilde n}{\wdw} - \centF{S}{\tilde n}{\wdw^*}}{\alpha}+ \omega > \emptF{S}{\tilde n}{\wdw^*} -\emptF{S}{\tilde n}{\wdw}\big] \\
    &\ge 1 - e^{-\kappa}\\
    \intertext{As $s \in W_{\ge t'}$, we know that $|\centF{S}{\tilde n}{\wdw} - \centF{S}{\tilde n}{\wdw^*}| \le 1-t'$.}
    \Leftrightarrow &\textstyle\Pr\big[ g(\kappa) + \frac{1-t'}{\alpha}+ \omega> \emptF{S}{\tilde n}{\wdw^*} - \emptF{S}{\tilde n}{\wdw} \big] 
    \ge 1 - e^{-\kappa}
\end{align*}
\end{proof}

\subsection{Proof of \cref{lem:prob_innerQuantileSelected_eventA}}
\label{SSec:UtilityProof_innerQuantile_eventA}
\begin{LemmaRepetition}{\ref{lem:prob_innerQuantileSelected_eventA}}
    Let $S\subseteq D$ be the current subset and $W$ the set of all split candidates and $W_{\ge t'}:= \{s \mid \centF{S}{\tilde n}{\wdw} \ge t'\}$ for any $t'\in [0,1]$.
    Let $e_{\min} := \min_{\wdw \in W} \emptF{S}{\tilde n}{\wdw}$ be the minimal emptiness over all splits $\wdw \in W$.
    The score of every split candidate can be represented as $\alpha \cdot e_{\min} + t' + \ln a_s$ (for some $a_s \ge 1$). Then with $
        L_{\ge t'} = \textstyle\sum_{\wdw \in W_{\ge t'}} a_\wdw 
    $ and $
        \textstyle L_{< t'}  = \sum_{\wdw \in W_{< t'}} a_s
    $, we know
    \begin{align*}
        &\Pr[\wdw \in W_{\ge t'}] = 
        \frac{%
                1%
            }{%
                \frac{%
                    L_{< t'}%
                }{L_{\ge t'}} + 1%
            }
        \text{.}
    \end{align*}
    \end{LemmaRepetition}
    
\begin{proof}
    To estimate the probability that the selected split \break $M_E(S,\score,\varepsilon)$ is in $W_{\ge t'}$, we derive bounds for the sum of the scores of the desired split candidates and over all possible split candidates.
    \begin{align*}
        &\Pr[M_E(S,\score,\varepsilon) \in W_{\ge t'}] = \frac{\sum_{\wdw \in W_{\ge t'}}\exp(\scoreF{S}{\tilde n}{\wdw}\varepsilon/(2\Delta_\score))}{\sum_{\wdw \in W}\exp(\scoreF{S}{\tilde n}{\wdw}\varepsilon/(2\Delta_\score))} 
    \end{align*}
    As we want a lower bound for this probability, we first characterise the numerator.
    By assumption, the minimal emptiness is $e_{\min}$
    and there is at least one split candidate with $\emptF{S}{\tilde n}{\wdw^*} > v$.  
    Note that we can represent all split candidates to have a score of $\alpha \cdot e_{\min} + t' + \ln a_s$. 
    We know that 
    \begin{align*}
        &\textstyle\sum_{\wdw \in W_{\ge t'}}\exp(\scoreF{S}{\tilde n}{\wdw}\varepsilon/(2\Delta_\score)) \\
        &\textstyle= \exp((\alpha \cdot e_{\min} +  t')\varepsilon/(2\Delta_\score))(\sum_{\wdw\in W_{\ge t'}} a_\wdw) \\
        \intertext{with $L_{\ge t'} := \sum_{\wdw\in W_{\ge t'}} a_\wdw$, $v = v' + e_{\min}$, and $P := \exp( (e_{\min}\alpha + t')\varepsilon/(2\Delta_\score))$ we get}
        &\textstyle= \exp((\alpha\cdot e_{\min} + t')\varepsilon/(2\Delta_\score)) L_{\ge t'}\\
        &\textstyle= PL_{\ge t'}
        \text{.}
    \end{align*}
    Next, we have to characterise the scores of all splits. 
    With $L_{< t'} := \sum_{\wdw\in W_{< t'}} a_\wdw$, we get
    \begin{align*}
        &\textstyle \sum_{\wdw \in W}\exp(\scoreF{S}{\tilde n}{\wdw}\varepsilon/(2\Delta_\score))\\ 
        &= P(L_{< t'} + L_{\ge t'})
    \end{align*}
    Putting this together, we get the following lower bound for the probability that a split in the inner quantile is selected.
    \begin{align*}
        \Pr[M_E(S,\score,\varepsilon)\in W_{\ge t'}] 
        &= \frac{%
                P L_{\ge t'}%
            }{%
                P(L_{< t'} + L_{\ge t'} )
            } \\
        &= \frac{%
                1%
            }{%
                \frac{%
                    L_{< t'}%
                }{L_{\ge t'}} + 1%
            }
   \end{align*}
\end{proof}

A better interpretable -- yet untight -- version of the above bound can be found by approximating $L_{\ge t'}$ and $L_{< t'}$. For the split candiates that have centreness less than $t'$, we know that $L_{< t'} \le |W_{< t'} \setminus W_{OPT_\omega}| \exp((OPT_\omega-t' -\alpha \empt_{\min})\varepsilon/(2\Delta_\score)) + |W_{< t'} \cap W_{OPT_\omega}|\exp((OPT-t'-\alpha \empt_{\min})\varepsilon/(2\Delta_\score))$. We know that there are $|W_{OPT_\omega}|$ splits with near-optimal scores. Then, $L_{\ge t'} \ge |W_{\ge t'}\setminus W_{OPT_\omega}|\exp(0) + |W_{OPT_\omega}\cap W_{\ge t'}| \exp((OPT_\omega-t' -\alpha \empt_{\min})\varepsilon/(2\Delta_\score))$
and $B=\varepsilon/(2\Delta_\score)$, which results in
    \begin{align*}
        &\Pr[M_E(S,\score,\varepsilon)\in W_{\ge t'}] \\
        &\ge \frac{%
                1%
            }{%
                \frac{%
                    e^{(OPT_\omega-t' -\alpha \empt_{\min})B}(|W_{< t'}\setminus W_{OPT_\omega}|  + |W_{< t'} \cap W_{OPT_\omega}|e^{\omega B})
                }{|W_{\ge t'}\setminus W_{OPT_\omega}| + | W_{\ge t'}\cap W_{OPT_\omega}| e^{(OPT_\omega -t'-\alpha \empt_{\min})B}} + 1%
            }\\
   \end{align*}

\subsection{Proof of \cref{lem:prob_innerQuantileSelected}}
\label{Ssec:UtilityProof_innerQuantileSelected}
\begin{LemmaRepetition}{\ref{lem:prob_innerQuantileSelected}}
    Let $S\subseteq D$ be the current set with noisy count $\tilde n$ and $W$ the set of all split candidates and $W_{\ge t'}:= \{s \mid \centF{S}{\tilde n}{\wdw} \ge t'\}$.
    Let $e_{\min} := \min_{\wdw \in W} \emptF{S}{\tilde n}{\wdw}$ be the minimal emptiness over all splits $\wdw \in W$.
    For any $t'\in [0,1]$, $0\le\eta\le |W|$, if $\eta|W_{\ge t'}|\exp((1+\alpha)\varepsilon/(2\Delta_\score)) \le |W_{< t'}| \exp((1-e_{min})\varepsilon/(2\Delta_\score))$ we have 
    \begin{align*}
        &\textstyle \Pr[\wdw \not\in W_{\ge t'}] \ge 
        1 - \frac{%
                1%
            }{%
                \eta + 1%
            }
        \text{,}
    \end{align*}
\end{LemmaRepetition}
\begin{proof}
The proof of 
    \[\Pr[\wdw \not\in W_{\ge t'}] = 
        1 - \frac{%
                1%
            }{%
                \frac{L_{< t'}}{L_{\ge t'}} + 1%
            }\]
    directly follows from the proof of \Cref{lem:utility_inner_quantile}. 
\end{proof}
To lower bound $\Pr[\wdw \not\in W_{\ge t'}]$, we need to upper bound $\Pr[\wdw \not\in W_{\ge t'}]$. Hence, we need to find an upper bound for $L_{< t'}$ and a lower bound for $L_{\ge t'}$. With $L_{< t'} \ge |W_{< t'}|\exp((1-e_{min})\alpha \varepsilon/(2\Delta_\score))$, and $L_{\ge t'} \le |W_{\ge t'}|\exp((1+ \alpha)\varepsilon/(2\Delta_\score))$ we get
    \begin{align*}
        &\Pr[\wdw \not\in W_{\ge t'}] =
        1 - \frac{%
                1%
            }{%
                L_{< t'}/L_{\ge t'} + 1%
            }\\
        & \ge 
        1 - \frac{%
                1%
            }{%
                \frac{
                    |W_{< t'}|\exp((1-e_{min})\alpha \varepsilon/(2\Delta_\score))
                }{
                    |W_{\ge t'}|\exp((1+ \alpha)\varepsilon/(2\Delta_\score))
                } + 1%
            }
        \intertext{If
        \[
            \eta|W_{\ge t'}|\exp((1+\alpha)\varepsilon/(2\Delta_\score)) \le |W_{< t'}| \exp((1-e_{min})\varepsilon/(2\Delta_\score))
        \] we get}
        & \ge
        1 - \frac{%
                1%
            }{%
                \eta + 1%
            }
\end{align*}

\section{Evaluation}

\subsection{Clustering Quality Metrics}
\label{app:metrics}

\begin{definition}[Inertia]
    \label{def:inertia}
    Given a data set $D = \{x_0, \dots, x_{n-1}\}$ and a set of cluster centres $C = \{c_0, \dots, c_{k-1}\}$ then the inertia is computed as follows: 
    \[
    I(D, C) = \sum_{i \in [k]} \sum_{x \in C_i} || c_i - x ||_2 
    \]
    and $|| . ||_2$ being the Euclidean norm.
\end{definition}
Minimising inertia requires finding a set of cluster centres that minimises the distance between each data point and its closest cluster centre. Another metric to assess the quality of a clustering result is the silhouette score:
\begin{definition}[Silhouette Score]
    \label{def:silhscore}
    Given a data set $D = \{x_0, \dots, x_{n-1}\}$ and a set of cluster centres $C = \{c_0, \dots, c_{k-1}\}$ then the silhouette score is computed as follows: 
    \[
    S(D, C) = \frac{1}{n} \sum_{i \in [k]} \sum_{x \in C_i} \frac{b(x) - a(x)}{\max(b(x), a(x))}
    \]
    with 
    \[
    a(x) = \frac{1}{|C_i| - 1} \sum_{x' \in C_i / \{x\}} || x' - x ||_2
    \]
    \[
    b(x) = \min_{j \neq i} \left( \frac{1}{|C_j|} \sum_{x' \in C_j} || x' - x ||_2 \right)
    \]
    .
\end{definition}
Note that the silhouette score is only defined for $|C| \geq 2$, otherwise it is set to $-1$.
In contrast to inertia, the silhouette score incorporates distances to other clusters and has to be maximised. Then, one has to evaluate whether it is worth to increase the number of clusters or not. In the case a data set is labelled, a clustering accuracy can also be computed:
\begin{definition}[Clustering Accuracy]
    \label{def:accuracy}
    Given a data set $D = \{x_0, \dots, x_{n-1}\}$ with corresponding labels $L = \{ l_0, \dots, l_{n-1} \}$ and a set of cluster centres $C = \{c_0, \dots, c_{k-1}\}$ then the accuracy is computed as follows: 
    \[
    A(D, L, C) = \frac{1}{n} \sum_{i \in [k]} \sum_{x_j \in C_i} \mathbbm{1}(\Call{lbl}{C_i} = y_j)  
    \]
    with
    \[
    \Call{lbl}{C} = \argmax_l |\{ x_j \in C \mid y_j = l \}|
    \]
    .
\end{definition}
The clustering accuracy assigns a label to each cluster based on the label of the majority of data points in that cluster and determines the proportion of matching labels.

\end{document}

%% file: commands.tex

\newcommand{\score}{f}
\newcommand{\scoreF}[3]{\score(#1,#2,#3)}
\newcommand{\scoreFull}[3]{\score_{t, q,\beta}(#1,#2,#3)}
\newcommand{\wdw}{s}
\newcommand{\cent}{c}
\newcommand{\centF}[3]{\cent(#1,#2,#3)}
\newcommand{\centFull}[3]{\cent_{t,q}(#1,#2,#3)}
\newcommand{\empt}{e}
\newcommand{\emptF}[3]{\empt(#1,#2,#3)}
\newcommand{\emptFull}[3]{\empt_{\beta}(#1,#2,#3)}
\newcommand{\rank}{r}
\newcommand{\rankF}[2]{\rank(#1,#2)}

\newcommand{\maxdepthRec}{{\tau_{r}}}
\newcommand{\minNumElem}{{\tau_{e}}}
\newcommand{\currDepthRec}{\gamma}
\newcommand{\datasets}{\mathbb{D}}
\newcommand{\outerQuantile}{Q_O}
\newcommand{\innerQuantile}{Q_I}
\newcommand{\numSplits}{numSplits}
\newcommand{\clusters}{cl}
\newcommand{\weights}{w}
\newcommand{\sigmas}{\Sigma}

\newcommand{\emptS}{e}
\newcommand{\centS}{c}
\newcommand{\ecSplit}{$(\emptS,\centS)$\xspace}
\newcommand{\eSplit}{$\emptS$ }

\newcommand{\emptweight}{\iota_{\empt}}
\newcommand{\centweight}{\iota_{\cent}}

\newcommand{\DPM}{\emph{DPM}\xspace}
\newcommand{\DPMOne}{\emph{DPM1}\xspace}
\newcommand{\DPMTwo}{\emph{DPM2}\xspace}

\newcommand{\R}{\mathbb{R}}
\newcommand{\Lap}{\text{Lap}}

\newcommand{\wdwxi}{\xi}

\newcommand{\epsCount}{\varepsilon_{\text{cnt}}}
\newcommand{\epsCounti}[1]{\varepsilon_{\text{cnt}_{#1}}}
\newcommand{\epsExpMech}{\varepsilon_{\text{exp}}}
\newcommand{\epsAverage}{\varepsilon_{\text{avg}}}
\newcommand{\epsWdwSize}{\varepsilon_{\text{int}}}
\newcommand{\deltaExpMech}{\delta_{\text{exp}}}
\newcommand{\deltaCount}{\delta_{\text{cnt}}}
\newcommand{\deltaAverage}{ \delta_{\text{avg}}}

\newcommand{\wdwSize}{\beta}

\newenvironment{LemmaRepetition}[1]{\textsc{Lemma #1.}\itshape}{}
\newenvironment{TheoremRepetition}[1]{\textsc{Theorem #1.}\itshape}{}

%% file: datasets.tex
\begin{table}[bp]
\centering
\caption{Summary of data sets used in our experiments.}
\label{tab:dts}
\begin{tabular}{@{}lllll@{}}
Name  & Data Points & Dimensions & Classes \\ \midrule
Synth-10d     & 100,000     & 10         & 64     \\
Synth-100d       & 100,000     & 100        & 64     \\
MNIST Embs.         & 60,000      & 40         & 10     \\
Fashion Embs. & 60,000      & 40         & 10     \\
UCI Letters        & 18,720      & 16         & 26     \\
UCI Gas  & 36,733      & 11         & (2)    \\ \bottomrule
\end{tabular}
\end{table}

%% file: main_arxiv_derived_from_final.bbl
\begin{thebibliography}{10}

\bibitem{balkan17}
M.-F. Balcan, T.~Dick, Y.~Liang, W.~Mou, and H.~Zhang.
\newblock Differentially private clustering in high-dimensional euclidean
  spaces.
\newblock In {\em International Conference on Machine Learning}, pages
  322--331. PMLR, 2017.

\bibitem{lshsplits2021}
A.~Chang, B.~Ghazi, R.~Kumar, and P.~Manurangsi.
\newblock Locally private k-means in one round.
\newblock In {\em ICML}, pages 1441--1451. PMLR, 2021.

\bibitem{RiskMinimization}
K.~Chaudhuri, C.~Monteleoni, and A.~D. Sarwate.
\newblock Differentially private empirical risk minimization.
\newblock {\em JMLR}, 12(3), 2011.

\bibitem{DPClEasyIns21}
E.~Cohen, H.~Kaplan, Y.~Mansour, U.~Stemmer, and E.~Tsfadia.
\newblock Differentially-private clustering of easy instances.
\newblock In {\em International Conference on Machine Learning}, pages
  2049--2059. PMLR, 2021.

\bibitem{expQ2020}
W.~Du, C.~Foot, M.~Moniot, A.~Bray, and A.~Groce.
\newblock Differentially private confidence intervals.
\newblock {\em arXiv preprint arXiv:2001.02285}, 2020.

\bibitem{Laplace}
C.~Dwork, F.~McSherry, K.~Nissim, and A.~Smith.
\newblock Calibrating noise to sensitivity in private data analysis.
\newblock In S.~Halevi and T.~Rabin, editors, {\em Theory of Cryptography},
  pages 265--284. Springer, 2006.

\bibitem{DwoRo14}
C.~Dwork, A.~Roth, et~al.
\newblock The algorithmic foundations of differential privacy.
\newblock {\em Foundations and Trends in Theoretical Computer Science}, 9,
  2014.

\bibitem{DPCl1MultCl17}
D.~Feldman, C.~Xiang, R.~Zhu, and D.~Rus.
\newblock Coresets for differentially private k-means clustering and
  applications to privacy in mobile sensor networks.
\newblock In {\em Proceedings of the 16th ACM/IEEE IPSN}, 2017.

\bibitem{DPClTightApprox20}
B.~Ghazi, R.~Kumar, and P.~Manurangsi.
\newblock Differentially private clustering: Tight approximation ratios.
\newblock {\em Neurips}, 2020.

\bibitem{repoGoogle}
Google.
\newblock Google's differential privacy libraries., 2021.

\bibitem{Optigrid99}
A.~Hinneburg and D.~A. Keim.
\newblock Optimal grid-clustering: Towards breaking the curse of dimensionality
  in high-dimensional clustering.
\newblock 1999.

\bibitem{diffprivlib}
N.~Holohan, S.~Braghin, P.~M.~Aonghusa, and K.~Levacher.
\newblock Diffprivlib: the {IBM} differential privacy library.
\newblock {\em ArXiv e-prints}, 1907.02444 [cs.CR], July 2019.

\bibitem{DPClOptDPWS18}
Z.~Huang and J.~Liu.
\newblock Optimal differentially private algorithms for k-means clustering.
\newblock In {\em Proceedings of the 37th ACM SIGMOD-SIGACT-SIGAI Symposium on
  Principles of Database Systems}, pages 395--408, 2018.

\bibitem{ConvexOpt}
R.~Iyengar, J.~P. Near, D.~Song, O.~Thakkar, A.~Thakurta, and L.~Wang.
\newblock Towards practical differentially private convex optimization.
\newblock In {\em IEEE S\&P}, pages 299--316. IEEE, 2019.

\bibitem{DPMaxCover21}
M.~Jones, H.~L. Nguyen, and T.~D. Nguyen.
\newblock Differentially private clustering via maximum coverage.
\newblock {\em Proceedings of the AAAI Conference on Artificial Intelligence},
  35(13), 2021.

\bibitem{mondrian06}
K.~LeFevre, D.~DeWitt, and R.~Ramakrishnan.
\newblock Mondrian multidimensional k-anonymity.
\newblock In {\em 22nd ICDE'06}, pages 25--25. IEEE, 2006.

\bibitem{Maddock}
S.~Maddock, G.~Cormode, T.~Wang, C.~Maple, and S.~Jha.
\newblock Federated boosted decision trees with differential privacy.
\newblock In {\em Proceedings of the 2022 ACM SIGSAC CCS}. ACM, 2022.

\bibitem{emmc21}
H.~L. Nguyen, A.~Chaturvedi, and E.~Z. Xu.
\newblock Differentially private k-means via exponential mechanism and max
  cover.
\newblock {\em Proceedings of the AAAI Conference on Artificial Intelligence},
  35(10):9101--9108, 5 2021.

\bibitem{DPUtilEffCl21}
T.~Ni, M.~Qiao, Z.~Chen, S.~Zhang, and H.~Zhong.
\newblock Utility-efficient differentially private k-means clustering based on
  cluster merging.
\newblock {\em Neurocomputing}, 424:205--214, 2021.

\bibitem{LocDP1Cl18}
K.~Nissim and U.~Stemmer.
\newblock Clustering algorithms for the centralized and local models.
\newblock In {\em Algorithmic Learning Theory}, pages 619--653. PMLR, 2018.

\bibitem{DPCl1Cluster16}
K.~Nissim, U.~Stemmer, and S.~Vadhan.
\newblock Locating a small cluster privately.
\newblock In {\em Proceedings of the 35th ACM SIGMOD-SIGACT-SIGAI Symposium on
  Principles of Database Systems}, pages 413--427, 2016.

\bibitem{scikit-learn}
F.~Pedregosa, G.~Varoquaux, A.~Gramfort, V.~Michel, B.~Thirion, O.~Grisel,
  M.~Blondel, P.~Prettenhofer, R.~Weiss, V.~Dubourg, J.~Vanderplas, A.~Passos,
  D.~Cournapeau, M.~Brucher, M.~Perrot, and E.~Duchesnay.
\newblock Scikit-learn: Machine learning in {P}ython.
\newblock {\em Journal of Machine Learning Research}, 12:2825--2830, 2011.

\bibitem{DPClStabAssump20}
M.~Shechner, O.~Sheffet, and U.~Stemmer.
\newblock Private k-means clustering with stability assumptions.
\newblock In {\em International Conference on Artificial Intelligence and
  Statistics}, pages 2518--2528. PMLR, 2020.

\bibitem{LocDPClKMeans}
U.~Stemmer.
\newblock Locally private k-means clustering.
\newblock {\em Journal of Machine Learning Research}, 22(176):1--30, 2021.

\bibitem{DPClConstMult18}
U.~Stemmer and H.~Kaplan.
\newblock Differentially private k-means with constant multiplicative error.
\newblock {\em Advances in Neural Information Processing Systems}, 31, 2018.

\bibitem{DPClLloyd16}
D.~Su, J.~Cao, N.~Li, E.~Bertino, and H.~Jin.
\newblock Differentially private k-means clustering.
\newblock In {\em Proceedings of the sixth ACM conference on data and
  application security and privacy}, pages 26--37, 2016.

\bibitem{PrivTree}
J.~Zhang, X.~Xiao, and X.~Xie.
\newblock Privtree: {A} differentially private algorithm for hierarchical
  decompositions.
\newblock {\em CoRR}, abs/1601.03229, 2016.

\bibitem{DPClPBDistContour18}
Y.~Zhang, N.~Liu, and S.~Wang.
\newblock A differential privacy protecting k-means clustering algorithm based
  on contour coefficients.
\newblock {\em PloS one}, 13(11):e0206832, 2018.

\end{thebibliography}
